\documentclass{amsart}%
\usepackage{amsmath}
\usepackage{amsfonts}
\usepackage{amssymb}
\usepackage{graphicx}%
\providecommand{\U}[1]{\protect\rule{.1in}{.1in}}
\newtheorem{theorem}{Theorem}

\newtheorem{lemma}[theorem]{Lemma}

\newtheorem{proposition}[theorem]{Proposition}

\begin{document}

\title[Segal--Bargmann transform for compact quotients]{The Segal--Bargmann transform for compact quotients of
symmetric spaces of the complex type}
\author{Brian C. Hall}
\address{University of Notre Dame\\Department of Mathematics\\Notre Dame, IN 46556-4618 USA}
\email{bhall@nd.edu}
\thanks{Supported in part by NSF grant DMS-0555862}
\author{Jeffrey J. Mitchell}
\address{Robert Morris University\\Department of Mathematics\\6001 University Boulevard\\Moon Township, PA 15108 USA}
\email{mitchellj@rmu.edu}

\begin{abstract}
Let $G/K$ be a Riemannian symmetric space of the complex type, meaning that
$G$ is complex semisimple and $K$ is a compact real form. Now let $\Gamma$ be
a discrete subgroup of $G$ that acts freely and cocompactly on $G/K.$ We
consider the Segal--Bargmann transform, defined in terms of the heat equation,
on the compact quotient $\Gamma\backslash G/K.$ We obtain isometry and
inversion formulas precisely parallel to the results we obtained previously
for globally symmetric spaces of the complex type. Our results are as parallel
as possible to the results one has in the dual compact case. Since there is no
known Gutzmer formula in this setting, our proofs make use of double coset
integrals and a holomorphic change of variable.

\end{abstract}

\maketitle


\section{Introduction}

\subsection{Segal--Bargmann transforms}

The Segal--Bargmann transform, in the form that we are considering in this
paper, consists of applying the heat operator to a function on a certain
Riemannian manifold $\mathcal{M}$ and then analytically continuing the result
to an appropriate complexification of $\mathcal{M}.$ An isometry formula shows
that the $L^{2}$ norm of the original function is equal to an appropriate norm
on the Segal--Bargmann transform, and an inversion formula shows how to
recover the original function from its Segal--Bargmann transform. So far, this
program has been carried out for Euclidean and compact Riemannian symmetric
spaces \cite{H1,H2,St1} and more recently for noncompact Riemannian symmetric
spaces \cite{KOS,OS2,HM3,HM4}, the Heisenberg group \cite{KTX}, and
nilmanifolds \cite{KTX2}.

The original motivation for this work came from quantum mechanics, in the work
of Segal \cite{Se1,Se2,Se3} and Bargmann \cite{Ba1}. The Segal--Bargmann
transform can be viewed as a sort of \textquotedblleft phase space wave
function\textquotedblright\ associated to the original \textquotedblleft
configuration space wave function.\textquotedblright\ Introduction of the
phase space wave function allows for several important new constructions,
including coherent states, the Berezin--Toeplitz quantization scheme, and the
Berezin transform.

On the other hand, one can consider the Segal--Bargmann transform as a
geometric study of the heat operator. From this point of view, the question
would be to try to characterize the range of the heat operator (for some fixed
time $t>0$). See \cite{range}. Since applying the heat operator always gives a
real-analytic function, it is natural to try to characterize functions in the
range in terms of appropriate conditions on their analytic continuations (the
isometry formula). Once the range of the heat operator is characterized, the
inversion formula is then a formula for computing the \textit{backward} heat
equation as an integral involving the analytic continuation of a function in
the range.

\subsection{Compact Lie groups}

We now review briefly the Segal--Bargmann transform for a compact Lie group
with a bi-invariant metric, which is a very special case of a compact
symmetric space. The main results of this paper are for compact quotients of
symmetric spaces of the \textquotedblleft complex type.\textquotedblright%
\ Noncompact symmetric spaces of the complex type are simply the duals (in the
usual duality for symmetric spaces) of compact Lie groups. The formulas for
the complex case, and compact quotients thereof, are very similar to those for
the compact group case, except that we will have to deal with
\textit{singularities}, which do not arise in the compact case.

Let $K$ be a compact Lie group and $K_{\mathbb{C}}$ its complexification. Let
$\Delta$ be the Laplacian on $K$ with respect to a bi-invariant metric, taken
to be a \textit{negative} operator, and let $e^{t\Delta/2}$ be the associated
(forward) heat operator. We fix $t>0$ and define the \textit{Segal--Bargmann
transform} for $K$ as the map taking $f\in L^{2}(K)$ to the holomorphic
extension to $K_{\mathbb{C}}$ of $e^{t\Delta/2}(f).$

For each $x\in K,$ the geometric exponential map $\exp_{x}:T_{x}(K)\rightarrow
K$ admits an extension to a holomorphic map of the complexified tangent space
$T_{x}(K)_{\mathbb{C}}$ into $K_{\mathbb{C}}.$ It can be shown (see \cite[Sec.
2]{St1} and also \cite[Sec. 8]{LGS}) that every point $z\in K_{\mathbb{C}}$
can be expressed uniquely as
\[
z=\exp_{x}iY,\quad x\in K,~Y\in T_{x}(K),
\]
where $\exp_{x}iY$ is defined by the just-described holomorphic extension. Let
$j_{x}$ be the Jacobian of $\exp_{x}$ and let%
\[
j_{x}^{\mathrm{nc}}(Y)=j_{x}(iY).
\]
The function $j_{x}^{\mathrm{nc}}$ may be thought of as the Jacobian of the
exponential mapping for the noncompact symmetric space dual to $K.$ (The
superscript \textquotedblleft$\mathrm{nc}$\textquotedblright\ stands for
\textquotedblleft noncompact.\textquotedblright) Let $\rho$ denote half the
sum (with multiplicities) of the positive restricted roots for $U/K.$ Then the
main results \cite{H1,H2} concerning the Segal--Bargmann transform for $K$ may
be described as follows.

\begin{theorem}
\label{group.thm}Let $K$ be a compact Lie group with a bi-invariant metric.
Then we have the following results.

The \textbf{isometry formula}. Fix $f$ in $L^{2}(K)$ and $t>0.$ Then
$F:=e^{t\Delta/2}f$ has a holomorphic extension to $K_{\mathbb{C}}$ satisfying%
\begin{align}
&  \int_{K}\left\vert f(x)\right\vert ^{2}~dx\nonumber\\
&  =e^{-\left\vert \rho\right\vert ^{2}t}\int_{x\in K}\int_{Y\in T_{x}%
(K)}\left\vert F(\exp_{x}iY)\right\vert ^{2}j_{x}^{\mathrm{nc}}(2Y)^{1/2}%
\frac{e^{-\left\vert Y\right\vert ^{2}/t}}{(\pi t)^{d/2}}~dY~dx.
\label{group.isom}%
\end{align}

The \textbf{surjectivity theorem}. Given any holomorphic function $F$ on
$K_{\mathbb{C}}$ for which the right-hand side of (\ref{group.isom}) is
finite, there exists a unique $f\in L^{2}(K)$ with $\left.  F\right\vert
_{K}=e^{t\Delta/2}f.$

The \textbf{inversion formula}. If $f\in L^{2}(K)$ is sufficiently regular and
$F:=e^{t\Delta/2}f,$ then%
\begin{equation}
f(x)=e^{-\left\vert \rho\right\vert ^{2}t/2}\int_{T_{x}(K)}F(\exp_{x}%
iY)j_{x}^{\mathrm{nc}}(Y)^{1/2}\frac{e^{-\left\vert Y\right\vert ^{2}/2t}%
}{(2\pi t)^{d/2}}~dY. \label{group.inv}%
\end{equation}

\end{theorem}

Note that we have $e^{-\left\vert Y\right\vert ^{2}/t}$ and $j_{x}%
^{\mathrm{nc}}(2Y)$ in the isometry formula but $e^{-\left\vert Y\right\vert
^{2}/2t}$ and $j_{x}^{\mathrm{c}}(Y)$ in the inversion formula. The same
results for the the Euclidean symmetric space $\mathbb{R}^{d}$ hold, with
$j_{x}^{\mathrm{nc}}(Y)\equiv1$ and $|\rho|=0.$ The $\mathbb{R}^{d}$ result is
the Segal--Bargmann transform for Euclidean spaces developed by Segal and
Bargmann \cite{Ba1,Se3}, with somewhat different normalization conventions.
(See \cite{range,mexnotes} for more information.)

For general compact symmetric spaces $U/K,$ the isometry and inversion
formulas developed by Stenzel \cite{St1} involve the \textit{heat kernel
measure} on the \textit{dual noncompact symmetric space}. In the case where
$U/K$ is isometric to a compact Lie group with a bi-invariant metric, the dual
noncompact symmetric space is of the \textquotedblleft complex
type,\textquotedblright\ where there is an explicit formula \cite{Ga} for the
heat kernel, accounting for the simple explicit form of Theorem
\ref{group.thm}.

The group case is also special because of connections to geometric
quantization \cite{geoquant,FMMN1,FMMN2,Hu} and the quantization of
$(1+1)$-dimensional Yang--Mills theory \cite{Wr,DH1,ymcoherent}.

\subsection{Quotients of noncompact symmetric spaces of the complex type}

In this paper, we consider a compact quotient of a noncompact Riemannian
symmetric space of the \textquotedblleft complex type.\textquotedblright%
\ Suppose $G$ is a connected \textit{complex} semisimple group and $K$ a
maximal compact subgroup of $G.$ Then the manifold $G/K$, equipped with a
fixed $G$-invariant Riemannian metric, is a noncompact symmetric space of the
complex type. Symmetric space of the noncompact type are nothing but the
noncompact duals, under the usual duality between symmetric spaces of compact
and noncompact types, of compact semisimple Lie groups. The simplest example
of a symmetric space of the complex type is hyperbolic 3-space.

In \cite{HM3,HM4}, we have developed a Segal--Bargmann transform for a
noncompact symmetric space $G/K$ of the complex type. (See also \cite{range}
further discussion of conceptual issues involved.) In the present paper, we
wish to extend that theory to a compact quotient $\Gamma\backslash G/K,$ where
$\Gamma$ is a discrete subgroup of $G$ acting freely and cocompactly on $G/K.$
(Examples of such quotients include compact hyperbolic 3-manifolds.) Although
the \textit{formulas} in the quotient case are essentially identical to the
formulas for $G/K$ itself, the \textit{proofs}, particularly of the isometry
formula, are different.

At the intuitive level, the results about the Segal--Bargmann transform for
$\Gamma\backslash G/K$ (with $G$ complex) should be obtained from the results
for the compact group case by dualizing. This means that we should replace $K$
in Theorem \ref{group.thm} with $\Gamma\backslash G/K$ and $j^{\mathrm{nc}}$
with $j^{\mathrm{c}},$ the Jacobian of the exponential map for the
\textit{compact} symmetric space dual to $G/K.$ We also replace $\left\vert
\rho\right\vert ^{2}$ with $-\left\vert \rho\right\vert ^{2}$ in the
exponential factors in front of the integrals, where $\left\vert
\rho\right\vert ^{2}$ is related to the scalar curvature, which is positive in
the compact case and negative in the noncompact case.

The challenge on the noncompact side (whether for $G/K$ or for $\Gamma
\backslash G/K$) is to make sense of the dualized formulas. The main
difficulty is the appearance of \textit{singularities} that do not appear on
the compact side. If $f$ is a function on $\Gamma\backslash G/K$ and we set
$F=e^{t\Delta/2}f$ for some fixed $t,$ then the function%
\[
Y\longmapsto F(\exp_{x}Y),\quad Y\in T_{x}(\Gamma\backslash G/K)
\]
does not admit an entire analytic holomorphic extension in $Y.$ Specifically, the function
$F(\exp_{x}(iY))$ will develop singularities once $Y$ gets large enough. (By
contrast, in the compact case, if $F$ is of the form $F=e^{t\Delta/2}f,$ then
$F(\exp_{x}(iY))$ is nonsingular for all $Y.$) To make sense of the isometry
formula or the inversion formula for $\Gamma\backslash G/K,$ we need a
\textit{cancellation of singularities}.

The inversion formula for $\Gamma\backslash G/K$ is as follows. Let $f$ be a
sufficiently smooth function in $L^{2}(\Gamma\backslash G/K)$ and let
$F=e^{t\Delta/2}f.$ Then we have
\begin{equation}
f(x)=\text{\textquotedblleft}\lim_{R\rightarrow\infty}\text{\textquotedblright%
}\ e^{\left\vert \rho\right\vert ^{2}t/2}\int_{\substack{Y\in T_{x}%
(\Gamma\backslash G/K) \\\left\vert Y\right\vert \leq R}}F(\exp_{x}%
iY)j_{x}^{\mathrm{c}}(Y)^{1/2}\frac{e^{-\left\vert Y\right\vert ^{2}/2t}%
}{(2\pi t)^{d/2}}~dY. \label{complex.inv1}%
\end{equation}
Here, \textquotedblleft$\lim_{R\rightarrow\infty}${}\textquotedblright\ means
the limit as $R$ tends to infinity of the\textit{\ real-analytic extension }of
the indicated quantity. That is to say, the integral on the right-hand side of
(\ref{complex.inv1}) is well-defined for all sufficiently small $R$ and admits
a real-analytic continuation in $R$ to $(0,\infty).$ The right-hand side of
(\ref{complex.inv1}) then is equal to the limit as $R$ tends to infinity of
this analytic continuation. (See also Stenzel's work \cite{St2} for a
different sort of inversion formula for noncompact symmetric spaces.)
Meanwhile, the isometry formula for $\Gamma\backslash G/K $ reads%
\begin{align}
&  \int_{\Gamma\backslash G/K}\left\vert f(x)\right\vert ^{2}~dx\nonumber\\
&  =\text{\textquotedblleft}\lim_{R\rightarrow\infty}\text{\textquotedblright%
}~e^{\left\vert \rho\right\vert ^{2}t}\int_{x\in\Gamma\backslash G/K}%
\int_{\substack{Y\in T_{x}(\Gamma\backslash G/K) \\\left\vert Y\right\vert
\leq R}}\left\vert F(\exp_{x}iY)\right\vert ^{2}j_{x}^{\mathrm{c}}%
(2Y)^{1/2}\frac{e^{-\left\vert Y\right\vert ^{2}/t}}{(\pi t)^{d/2}}~dY~dx.
\label{complex.isom1}%
\end{align}

In both the inversion formula and the isometry formula, there is a
cancellation of singularities that allows the real-analytic extension with
respect to $R$ to exist, even though $F(\exp_{x}iY)$ becomes singular for
large $Y.$ In the inversion formula, for example, integral on the right-hand
side of (\ref{complex.inv1}) is unchanged if we average the function
$Y\longmapsto F(\exp_{x}iY)$ with respect to the action of $K_{x},$ the group
of local isometries of $\Gamma\backslash G/K$ fixing $x.$ This averaging
process cancels many of the singularities in $F(\exp_{x}iY)$; the remaining
singularities are canceled by the zeros of the function $j^{\mathrm{c}}(Y).$

Our inversion and isometry formulas for the Segal--Bargmann transform on
$\Gamma\backslash G/K$ are the same as the ones developed in \cite{HM3,HM4}
for $G/K,$ except for replacing $G/K$ with $\Gamma\backslash G/K$ in the
obvious places in the formulas. (In both cases, still assuming that $G$ is
complex!) In the $G/K$ case, our isometry formula \textit{does not} coincide
with the isometry formula developed by B. Kr\"{o}tz, G. \'{O}lafsson, and R.
Stanton \cite{KOS}. The results of \cite{KOS} have the advantage of working
for arbitrary symmetric spaces of the noncompact type (not just the complex
case); our results, meanwhile, have the advantage of being more parallel to
what one has in the compact case. (See \cite{OS1} a Segal--Bargmann transform
for radial functions on noncompact symmetric spaces. See also \cite{OS2} for a
refinement of the isometry formula in \cite{KOS}, which also differs from the
isometry formula of \cite{HM4} when specialized to the complex case.)

\subsection{Remarks on the methods used}

We conclude this introduction by discussing the methods of proof. At least
conceptually, the proof of the inversion formula on $\Gamma\backslash G/K$
should be similar to the proof of the inversion formula on $G/K.$ After all, a
function $f$ on $\Gamma\backslash G/K$ lifts to a $\Gamma$-invariant function
$\tilde{f}$ on $G/K.$ To be sure, $\tilde{f}$ is not square-integrable on
$G/K,$ but this matters little, since the inversion formula involves no
integration over the base manifold. As a result, our proof of the inversion
formula for $\Gamma\backslash G/K$ is similar to the proof for $G/K.$ The key
ingredient is an \textquotedblleft intertwining formula,\textquotedblright%
\ specific to the complex case, between the Euclidean and non-Euclidean Laplacians.

In the case of the isometry formula, lifting to $G/K$ is not helpful, since
the lack of square-integrability of $\tilde{f}$ prevents us from formulating
the isometry \textquotedblleft upstairs\textquotedblright\ on $G/K.$
Meanwhile, the Gutzmer-type formula of J. Faraut \cite{Far1,Far2}, which is
the key ingredient in the proofs in \cite{HM3,HM4} and also in \cite{KOS}, has
no analog (so far as we know) on $\Gamma\backslash G/K.$ This means that our
proof of the isometry formula for $\Gamma\backslash G/K$ must use methods that
are completely different from those in the \cite{HM3,HM4}. Our proof uses a
double coset integral along with a holomorphic change of variable to reduce
the isometry formula to the inversion formula. This approach parallels one
method of establishing the isometry formula in the compact case, in work of
Hall \cite{H2} and Stenzel \cite{St1}.

The proof of the holomorphic change of variable (Theorem \ref{holo_change.thm}%
) applies to compact quotients of general symmetric spaces of the noncompact
type, not just those of the complex type. Meanwhile, in cases where the
relevant singularities can be understood fairly easily (say, the rank-one case
or the even-multiplicity case), it may be possible to develop an inversion
formula involving integration against a suitably \textquotedblleft
unwrapped\textquotedblright\ version of the heat kernel on the dual compact
symmetric space. Thus, it may be possible to develop results similar to those
of this paper and \cite{HM3,HM4} for other noncompact symmetric spaces and
compact quotients thereof.

\section{Set-up\label{setup.sec}}

We begin this section by recalling certain basic facts about symmetric spaces.
A standard reference for this material is \cite{He1}. We consider a connected
\textit{complex} semisimple group $G$, a fixed maximal compact subgroup $K$ of
$G,$ and the quotient manifold $G/K.$ We will assume, with no loss of
generality, that $G$ acts effectively on $G/K,$ which is equivalent to
assuming that the Lie algebra $\mathfrak{k}$ of $K$ contains no nonzero ideal
of $\mathfrak{g}$ and that the center of $G$ is trivial. There is then a
unique involution of $G$ whose fixed-point subgroup is $K.$ The Lie algebra
$\mathfrak{g}$ then decomposes as $\mathfrak{g}=\mathfrak{k}+\mathfrak{p},$
where $\mathfrak{p}$ is the space on which the associated Lie algebra
involution acts as $-I.$ (Since $G$ is complex, $\mathfrak{p}$ will be equal
to $i\mathfrak{k}.$) We now choose on $\mathfrak{p}$ an inner product
invariant under the adjoint action of $K.$ We consider the manifold $G/K$ and
we identify the tangent space at the identity coset $x_{0}$ with
$\mathfrak{p}.$ There is then a unique $G$-invariant Riemannian structure on
$G/K$ whose restriction to $T_{x_{0}}(G/K)=\mathfrak{p}$ is the chosen
Ad-$K$-invariant inner product. The manifold $G/K$, together with a Riemannian
structure of this form, is what we will call a noncompact symmetric space of
the complex type.

We let $\mathfrak{a}$ be a maximal commutative subspace of $\mathfrak{p}$ and
we let $R\subset\mathfrak{a}^{\ast}$ denote the set of (restricted) roots for
the $(\mathfrak{g},\mathfrak{k}).$ We fix a set of positive roots, which we
denote by $R^{+}.$ We then let $\mathfrak{a}^{+}$ denote the closed
fundamental Weyl chamber, that is, the set of $Y\in\mathfrak{a}$ such that
$\alpha(Y)\geq0$ for all $\alpha\in R^{+}.$ It is known that every element of
$\mathfrak{p}$ can be moved into $\mathfrak{a}^{+}$ by the adjoint action of
$K.$

We will also consider the \textit{compact dual} to $G/K.$ Let $G_{\mathbb{C}}
$ be the complexification of $G,$ which contains $G$ as a closed subgroup, in
which case the Lie algebra of $G_{\mathbb{C}}$ is $\mathfrak{g}_{\mathbb{C}%
}=\mathfrak{g}+i\mathfrak{g}.$ We define $\mathfrak{u}$ to be the subalgebra
of $\mathfrak{g}_{\mathbb{C}}$ given by $\mathfrak{u}:=\mathfrak{k}%
+i\mathfrak{p}$ and we let $U$ be the corresponding connected Lie subgroup of
$G_{\mathbb{C}},$ which is compact. We then consider the manifold $U/K.$ We
think of the tangent space at the identity coset in $U/K$ as $i\mathfrak{p}$.
The chosen inner product on $\mathfrak{p}$ then determines an inner product on
$i\mathfrak{p}$ in the obvious way. There is then a unique $U$-invariant
Riemannian structure on $U/K$ whose restriction to the tangent space at the
identity coset is this inner product. The manifold $U/K$, with this Riemannian
structure, is a simply connected symmetric space of the compact type in the
notation of \cite{He1}, and is called the compact dual of the symmetric space
$G/K.$ Since $G/K$ is of the complex type, $U/K$ will be isometric to a
compact Lie group with a bi-invariant metric.

We then consider a discrete subgroup $\Gamma$ of $G$ with the property that
$\Gamma$ acts freely and cocompactly on $G/K.$ The action of $\Gamma$ is then
automatically properly discontinuous. It is not obvious but true that such
subgroups always exist. The manifold $X:=\Gamma\backslash G/K$ is then what we
mean by a compact quotient of $G/K.$ We let $\pi$ denote the quotient map from
$G/K$ to $\Gamma\backslash G/K$; this map is a covering map. Because the
action of $\Gamma\subset G$ on $G/K$ is isometric, the metric on $G/K$
descends unambiguously to $X.$ In the case that $G/K$ is hyperbolic 3-space, a
compact quotient is nothing but a hyperbolic 3-manifold, that is, an
orientable closed 3-manifold of constant negative curvature.

For $R>0,$ let $T^{R}(X)$ denote the set of pairs $(x,Y)$ in $T(X)$ with
$\left\vert Y\right\vert <R.$ Let $S_{R}$ denote the strip in the complex
plane given by%
\begin{equation}
S_{R}=\left\{  \left.  u+iv\in\mathbb{C}\right\vert ~\left\vert v\right\vert
<R\right\}  . \label{sr}%
\end{equation}
If $\gamma$ is a unit-speed geodesic in $X,$ consider the map $\tau
:S_{R}\rightarrow T^{R}(X)$ given by
\[
\tau(u+iv)=(\gamma(u),v\dot{\gamma}(u)).
\]
In the terminology of Lempert and Sz\H{o}ke \cite{LS,Sz1}, a complex structure
on $T^{R}(X)$ is called \textquotedblleft adapted\textquotedblright\ (to the
given metric on $X$) if for each geodesic $\gamma,$ the map $\tau$ is
holomorphic as a map of $S_{R}\subset\mathbb{C}$ into $T^{R}(X).$ Lempert and
Sz\H{o}ke show that for any $R>0$ there exists at most one adapted complex
structure and that if $R$ is small enough then an adapted complex structure
does exist. (These results hold more generally for any compact, real-analytic
Riemannian manifold.) The same complex structure was constructed
independently, from a different but equivalent point of view, by Guillemin and
Stenzel \cite{GStenz1,GStenz2}.

Given $x\in X,$ we may consider the geometric exponential map
\[
\exp_{x}:T_{x}(X)\rightarrow X\subset T(X).
\]
This map can be analytically continued into a holomorphic map of a
neighborhood of the identity in the complexified tangent space $T_{x}%
(X)_{\mathbb{C}}$ into $T(X).$ This analytically continued exponential map
satisfies%
\begin{equation}
\exp_{x}(iY)=(x,Y), \label{exp.id}%
\end{equation}
as may easily be verified from the holomorphicity of the map $\tau.$

If $F$ is a real-analytic function on $X,$ it will have an analytic
continuation, also called $F,$ to some $T^{R^{\prime}}(X),$ for some
$R^{\prime}\leq R.$ In light of (\ref{exp.id}), we may write the value of $F$
at a point $(x,Y)\in T^{R^{\prime}}(X)$ as $F(\exp_{x}iY).$ This notation is
suggestive, because we may alternatively consider the map%
\begin{equation}
Y\rightarrow F(\exp_{x}Y) \label{y_f}%
\end{equation}
as a real-analytic map of $T_{x}(X)$ into $\mathbb{C}.$ Then the expression
$F(\exp_{x}iY)$ may be thought of equivalently as the analytic continuation of
$F$ evaluated at the point $\exp_{x}(iY)=(x,Y),$ or as the analytic
continuation of the map (\ref{y_f}), evaluated at the point $iY.$

For each $x\in X,$ we have also the Jacobian $j_{x}$ of the exponential map
$\exp_{x}.$ To compute $j_{x}$, we choose some $\tilde{x}\in G/K$ that maps to
$x$ under quotienting by $\Gamma$. Then we choose some $g\in G$ with $g\cdot
x_{0}=\tilde{x},$ where $x_{0}$ is the identity coset in $G/K.$ The action of
$g$ serves to identify $T_{x_{0}}(G/K)$ with $T_{\tilde{x}}(G/K),$ which is
then naturally identifiable with $T_{x}(X)$ by the differential of the
covering map from $G/K$ to $\Gamma\backslash G/K$. Finally, $T_{x_{0}}(G/K)$
is naturally identifiable with $\mathfrak{p}.$ In this way, we obtain an
identification of $T_{x}(X)$ with $\mathfrak{p}.$ The identification is not
unique, but it is unique up to the adjoint action of $K $ on $\mathfrak{p}.$
Under any identification of this sort, $j_{x}$ is invariant under the adjoint
action of $K$ on $\mathfrak{p},$ and the restriction of $j_{x}$ to
$\mathfrak{a}\subset\mathfrak{p}$ is given by%
\[
j_{x}(H)=\prod_{\alpha\in R^{+}}\left(  \frac{\sinh\alpha(H)}{\alpha
(H)}\right)  ^{2}.
\]
This formula is the same as the formula for the Jacobian of the exponential
map on $G/K$ and reflects that in the complex case, all the (restricted) roots
for $G/K$ have multiplicity 2.

From this formula, one can verify that the function $j_{x}$ on $T_{x}(X)$
admits an entire holomorphic extension to the complexification of $T_{x}(X),$
which may be identified with $\mathfrak{p}_{\mathbb{C}}.$ Now consider the
function $j_{x}^{\mathrm{c}}$ on $T_{x}(X)$ given by%
\[
j_{x}^{\mathrm{c}}(Y)=j_{x}(iY).
\]
Under our identification of $T_{x}(X)$ with $\mathfrak{p},$ we have that
$j_{x}^{\mathrm{c}}$ is invariant under the adjoint action of $K$ and its
restriction to $\mathfrak{a}$ is given by%
\[
j_{x}^{\mathrm{c}}(H)=\prod_{\alpha\in R^{+}}\left(  \frac{\sin\alpha
(H)}{\alpha(H)}\right)  ^{2}.
\]
The superscript \textquotedblleft$\mathrm{c}$\textquotedblright\ in the
formula reflects that $j_{x}^{\mathrm{c}}$ coincides with the Jacobian of the
exponential mapping for the \textit{compact} symmetric space $U/K.$

Note that the formula for $j_{x}^{\mathrm{c}}$, under any identification of
$T_{x}(X)$ with $\mathfrak{p}$ of the above sort, is independent of $x.$ Thus,
in a certain sense, $j_{x}^{\mathrm{c}}$ is \textquotedblleft the
same\textquotedblright\ function for each $x,$ reflecting that any point in
$X$ can be mapped to any other point by a local isometry. For example, in the
case of a hyperbolic 3-manifold (with an appropriate normalization of the
metric), we have $j_{x}^{\mathrm{c}}(Y)=(\sin\left\vert Y\right\vert
/\left\vert Y\right\vert )^{2}$ for every $x\in X.$

We now let $j_{x}^{\mathrm{c}}(Y)^{1/2}$ be (under our identification of
$T_{x}(X)$ with $\mathfrak{p}$) the Ad-$K$-invariant function whose
restriction to $\mathfrak{a}$ is given by%
\begin{equation}
j_{x}^{\mathrm{c}}(H)^{1/2}=\prod_{\alpha\in R^{+}}\frac{\sin\alpha(H)}%
{\alpha(H)}. \label{jc_half}%
\end{equation}
Note that $j_{x}^{\mathrm{c}}(Y)^{1/2}$ is \textit{not} the positive square
root of $j_{x}^{\mathrm{c}}(Y).$ Rather, $j_{x}^{\mathrm{c}}(Y)^{1/2}$ is
chosen so as to be \textit{real analytic} and positive near the origin. We
then let $j_{x}^{\mathrm{c}}(Y)^{-1/2}$ be the reciprocal of $j_{x}%
^{\mathrm{c}}(Y)^{1/2},$ defined away from the points where $j_{x}%
^{\mathrm{c}}(Y)$ is zero.

\section{The inversion formula\label{inv.sec}}

The key result of this section is the partial inversion formula (Theorem
\ref{partial.thm}), which is proved using an intertwining formula that relates
the Laplacian on $G/K$ to the Euclidean Laplacian. Once Theorem
\ref{partial.thm} is proved, the desired \textquotedblleft
global\textquotedblright\ inversion formulas follow by a fairly
straightforward limit as the radius tends to infinity.

For each $x\in\Gamma\backslash G/K,$ consider (as in \cite{HM3,HM4}) the
function
\[
\nu_{t,x}^{\mathrm{c}}(Y):=e^{t\left\vert \rho\right\vert ^{2}/2}%
j_{x}^{\mathrm{c}}(Y)^{-1/2}\frac{e^{-\left\vert Y\right\vert ^{2}/2t}}{(2\pi
t)^{d/2}},\quad Y\in T_{x}(X)
\]
and the associated signed measure%
\begin{equation}
\nu_{t,x}^{\mathrm{c}}(Y)j_{x}^{\mathrm{c}}(Y)~dY=e^{t\left\vert
\rho\right\vert ^{2}/2}j_{x}^{\mathrm{c}}(Y)^{1/2}\frac{e^{-\left\vert
Y\right\vert ^{2}/2t}}{(2\pi t)^{d/2}}~dY. \label{nu_t_c}%
\end{equation}
Here, again, the superscript \textquotedblleft$\mathrm{c}$\textquotedblright%
\ is supposed to denote quantities associated to the compact symmetric space
$U/K$ dual to $G/K.$ The measure in (\ref{nu_t_c}) is an \textquotedblleft
unwrapped\textquotedblright\ form of the heat kernel measure on $U/K.$ This
means that the push-forward of this measure under the exponential mapping for
$U/K$ is precisely the heat kernel measure at the identity coset on $U/K$
\cite[Thm. 5]{HM3}. (Note that because $U/K$ is isometric to a compact Lie
group with a bi-invariant metric, the heat kernel formula of \`{E}skin
(\cite{E}; see also \cite{U}) applies. From this formula it is easy to see
that the signed measure in (\ref{nu_t_c}) pushes forward to the heat kernel on
$U/K.$)

We introduce the operator%
\begin{align}
&  A_{t,R}(f)(x)=\int_{T_{x}^{R}(X)}F(\exp_{x}iY)\nu_{t,x}^{\mathrm{c}%
}(Y)j_{x}^{\mathrm{c}}(Y)~dY\nonumber\\
&  =e^{t\left\vert \rho\right\vert ^{2}/2}\int_{T_{x}^{R}(X)}F(\exp
_{x}iY)j_{x}^{\mathrm{c}}(Y)^{1/2}\frac{e^{-\left\vert Y\right\vert ^{2}/2t}%
}{(2\pi t)^{d/2}}~dY, \label{a_def}%
\end{align}
where as usual, $F$ is the analytic continuation of $e^{t\Delta/2}f$ and where
$T_{x}^{R}(X)$ denotes the vectors in $T_{x}(X)$ with length less than $R.$
The operator $A_{t,R}$ consists of applying the time-$t$ heat operator and
then doing a \textquotedblleft partial inversion,\textquotedblright\ in which
we integrate only over a ball of radius $R$ in the tangent space. We will seek
a way to allow $R$ to tend to infinity, by means of an appropriate analytic
continuation, with the expectation that $A_{t,R}$ tends to the identity
operator as $R$ tends to infinity.

We now state the results of this section, before turning to the proofs.

\begin{proposition}
\label{continue.prop}There exists $R_{0}>0$ such that for all $f\in L^{2}(X),
$ the function $F:=e^{t\Delta/2}f$ has a holomorphic extension to $T^{R_{0}%
}(X),$ with respect to the adapted complex structure. Furthermore, for each
fixed $z\in T^{R_{0}}(X),$ the map $f\rightarrow F(z)$ is a bounded linear
functional on $L^{2}(X),$ with norm a locally bounded function of $z.$
\end{proposition}

This proposition shows that the operator $A_{t,R}$ is well defined and bounded
for all sufficiently small $R.$

\begin{theorem}
[Partial Inversion Formula]\label{partial.thm}Let $R_{0}$ be as in Proposition
\ref{continue.prop}. For all $R<R_{0},$ let $A_{t,R}$ be the operator defined
by (\ref{a_def}). Then $A_{t,R}$ is a bounded operator on $L^{2}(X)$ and is
given by%
\[
A_{t,R}=\alpha_{t,R}(-\Delta),
\]
where $\alpha_{t,R}:[0,\infty)\rightarrow\mathbb{R}$ is given by%
\begin{equation}
\alpha_{t,R}(\lambda)=e^{-t\lambda/2}e^{t\left\vert \rho\right\vert ^{2}%
/2}\int_{\substack{Y\in\mathbb{R}^{d} \\\left\vert Y\right\vert \leq R }%
}\exp\left(  \sqrt{\lambda-\left\vert \rho\right\vert ^{2}}~y_{1}\right)
\frac{e^{-\left\vert Y\right\vert ^{2}/2t}}{(2\pi t)^{d/2}}~dY.
\label{alpha.int}%
\end{equation}
Here $\sqrt{\lambda-\left\vert \rho\right\vert ^{2}}$ is either of the two
square roots of $\lambda-\left\vert \rho\right\vert ^{2}.$
\end{theorem}

On $G/K,$ the spectrum of $-\Delta$ is the interval $[\left\vert
\rho\right\vert ^{2},\infty).$ By contrast, on $X=\Gamma\backslash G/K,$ the
spectrum of $-\Delta$ includes points in the interval $[0,\left\vert
\rho\right\vert ^{2})$; for example, the constant function $\mathbf{1}$ is an
eigenvector for $-\Delta$ with eigenvalue $0.$ For $\lambda\in\lbrack
0,\left\vert \rho\right\vert ^{2}),$ $\sqrt{\lambda-\left\vert \rho\right\vert
^{2}}$ will be pure imaginary. Nevertheless, because the domain of integration
in (\ref{alpha.int}) is invariant under $y\rightarrow-y,$ the value of
$\alpha_{t,R}(\lambda)$ is still a real number.

For each fixed value of $R$ and $t,$ the integral in (\ref{alpha.int}) is
bounded by a constant times $\exp(\sqrt{\lambda}R).$ Thus, because of the
factor of $e^{-t\lambda/2}$ in front of the integral, $\alpha_{t,R}(\lambda)$
is a bounded function of $\lambda$ for each $R$ and $t.$

Although the definition of $A_{t,R}$ in (\ref{a_def}) makes sense only for
small $R,$ the function in (\ref{alpha.int}) is a well defined and bounded
function of $\lambda$ for every $R>0.$ Furthermore, if we let $R$ tend to
infinity in the definition of $\alpha_{t,R}$ we obtain (by Dominated
Convergence) an integral over all of $\mathbb{R}^{d}.$ This integral is an
easily evaluated Gaussian integral, whose value turns out to be 1 for all
$\lambda.$ That is to say,%
\begin{equation}
\lim_{R\rightarrow\infty}\alpha_{t,R}(\lambda)=1 \label{alpha.lim}%
\end{equation}
for all $t$ and $\lambda.$ This suggests that $A_{t,R}(f)$ should tend to $f $
as $R$ tends to infinity; proving this will yield a global inversion formula.
We present two versions of the formula, an $L^{2}$ version valid for all $f\in
L^{2}(X)$ and a pointwise version valid for sufficiently smooth $f.$

\begin{theorem}
[Global Inversion Formula, $L^{2}$ Version]\label{global_l2.thm}Let $R_{0}$ be
as in Proposition \ref{continue.prop}. For all $R<R_{0},$ let $A_{t,R}$ be as
in (\ref{a_def}). Then the map $R\rightarrow A_{t,R}$ has a weakly analytic
extension, also denoted $A_{t,R}, $ to a map of $(0,\infty)$ into the space of
bounded operators on $L^{2}(X). $ This analytic extension has the property
that for each $f\in L^{2}(X)$ we have%
\begin{equation}
f=\lim_{R\rightarrow\infty}A_{t,R}f, \label{global.inv1}%
\end{equation}
with the limit being in the norm topology of $L^{2}(X).$ In light of the
original expression for $A_{t,R},$ we may express (\ref{global.inv1})
informally as%
\[
f(x)=\text{\textquotedblleft}\lim_{R\rightarrow\infty}\text{{}%
\textquotedblright}e^{t\left\vert \rho\right\vert ^{2}/2}\int_{T_{x}^{R}%
(X)}F(\exp_{x}iY)j_{x}^{\mathrm{c}}(Y)^{1/2}\frac{e^{-\left\vert Y\right\vert
^{2}/2t}}{(2\pi t)^{d/2}}~dY,
\]
with the limit in the $L^{2}$ sense.
\end{theorem}

Recall that a map $\alpha$ of $(0,\infty)$ into the space of bounded operators
on a Hilbert space $H$ is weakly analytic if the map $R\rightarrow\left\langle
f,\alpha(R)g\right\rangle $ is a real-analytic function of $R$ for each $f$
and $g$ in $H.$ Of course, the analytic extension of the map $R\rightarrow
A_{t,R}$ is given by $\alpha_{t,R}(-\Delta),$ where $\alpha_{t,R}$ is defined
(for all $R>0$) by (\ref{alpha.int}).

\begin{theorem}
[Global Inversion Formula, Pointwise Version]\label{global_ptwise.thm}Let
$R_{0}$ be as in Proposition \ref{continue.prop}. For all $R<R_{0},$ let
$A_{t,R}$ be as in (\ref{a_def}). Assume that $f\in L^{2}(X)$ is in the domain
of $\Delta^{l}$ for some positive real number with $l>(3d^{2}-d)/4.$ Then for
each $x\in X,$ the function $L_{x,f}(F)$ given by%
\[
L_{x,f}(R)=(A_{t,R}f)(x)
\]
has a real-analytic extension, also denoted $L_{x,f},$ from $R\in(0,R_{0})$ to
$R\in(0,\infty).$ Furthermore, we have
\begin{subequations}
\begin{equation}
f(x)=\lim_{R\rightarrow\infty}L_{x,f}(R), \label{global.inv2}%
\end{equation}
with the limit being uniform in $x.$ In light of the original expression for
$A_{t,R},$ we may express (\ref{global.inv1}) informally as
\end{subequations}
\[
f(x)=\text{\textquotedblleft}\lim_{R\rightarrow\infty}{}%
\text{\textquotedblright}e^{t\left\vert \rho\right\vert ^{2}/2}\int_{T_{x}%
^{R}(X)}F(\exp_{x}iY)j_{x}^{\mathrm{c}}(Y)^{1/2}\frac{e^{-\left\vert
Y\right\vert ^{2}/2t}}{(2\pi t)^{d/2}}~dY,
\]
for $f$ in $Dom(\Delta^{l})$, with the limit being uniform in $x.$
\end{theorem}

These inversion formulas are as parallel as possible to the inversion formula
(\ref{group.inv}) in the dual compact group case. Specifically, the inversion
formulas above are obtained by \textquotedblleft dualizing\textquotedblright%
\ (\ref{group.inv}) (changing $j_{x}^{\mathrm{nc}} $ to $j_{x}^{\mathrm{c}}$
and $e^{-t\left\vert \rho\right\vert ^{2}/2}$ to $e^{t\left\vert
\rho\right\vert ^{2}/2}$) and inserting an analytic continuation in $R,$ which
is unnecessary in the compact group case.

\begin{proof}
(Proof of Proposition \ref{continue.prop}) Let $k_{t}(\cdot,\cdot)$ denote the
heat kernel for $X.$ A result of Nelson \cite[Thm. 8]{N} shows that for any
fixed positive time $t,$ $k_{t}$ is a real-analytic function on $X\times X.$
As a result, $k_{t}$ will have an analytic continuation, also denoted $k_{t},$
to $T^{R_{0}}(X)\times T^{R_{0}}(X)$ for some sufficiently small $R_{0}.$ Then
the function defined by%
\[
z\rightarrow\int_{X}k_{t}(z,y)f(y)~dy,\quad z\in T^{R_{0}}(X),
\]
is the desired holomorphic extension of $F:=e^{t\Delta/2}f.$ The desired
properties of the pointwise evaluation functional are then easy to read off.

Now, Nelson's result leaves open the possibility that the radius $R_{0}$ could
depend on $t,$ which is harmless in our case, since we work with one fixed $t$
throughout. Nevertheless, using a result of Guillemin and Stenzel \cite[Thm.
5.2]{GStenz2}, it is not hard to see that $R_{0}$ can be chosen to be
independent of $t.$
\end{proof}

\begin{proof}
(Proof of Theorem \ref{partial.thm}) Now that $A_{t,R}$ is known to be a
bounded operator, we can compute it by evaluating it on an orthonormal basis
for $L^{2}(X)$ consisting of eigenfunctions of the Laplacian. So let $\phi$ be
an eigenfunction of $-\Delta$ on $X$ with eigenvalue $\lambda\geq0.$ Our goal
is to show that $A_{t,R}(\phi)$ is a certain constant multiple of $\phi,$ with
the constant depending only on $\lambda.$ This will show that $A_{t,R}$ is a
specific function of the Laplacian.

Applying $e^{t\Delta/2}$ to $\phi$ gives $e^{-t\lambda/2}\phi.$ This means
that we want to compute%
\[
e^{-t\lambda/2}\int_{T_{x}^{R}(X)}\phi(\exp_{x}iY)\nu_{t,x}^{\mathrm{c}%
}(Y)j_{x}^{\mathrm{c}}(Y)~dY,
\]
for a fixed $x$ in $X=\Gamma\backslash G/K.$ Let $K_{x}$ denote the identity
component of the group of \textit{local} isometries of $X$ that fix $x.$ Then
the key point is that the function $\nu_{t,x}^{\mathrm{c}}(Y)j_{x}%
^{\mathrm{c}}(Y)$ is invariant under the action of $K_{x}$ on $T_{x}(X).$
(This invariance can be seen by observing that under our identification of
$T_{x}(X)$ with $\mathfrak{p},$ $j_{x}^{\mathrm{c}}$ is an Ad-$K$-invariant
function on $\mathfrak{p}.$) Thus averaging the function $Y\longmapsto
\phi(\exp_{x}iY)$ over the action of $K_{x}$ has no effect on the integral.
This averaging cancels out many of the singularities in $\phi(\exp_{x}iY).$

Let $\tilde{x}$ be a preimage of $x$ in $G/K$ and let $\Phi$ be the lift of
$\phi$ to $G/K.$ Let $\Phi^{(\tilde{x})}$ denote the radialization of $\Phi$
about $\tilde{x},$ that is, the average of $\Phi$ over the action of
$K_{\tilde{x}},$ where $K_{\tilde{x}}$ is the stabilizer of $\tilde{x}$ in
$G.$ Because $\nu_{t,x}^{\mathrm{c}}(Y)j_{x}^{\mathrm{c}}(Y)$ is invariant
under the action of $K_{x},$ we have
\begin{align}
&  e^{-t\lambda/2}\int_{T_{x}^{R}(X)}\phi(\exp_{x}iY)\nu_{t,x}^{\mathrm{c}%
}(Y)j_{x}^{\mathrm{c}}(Y)~dY\nonumber\\
&  =e^{-t\lambda/2}\int_{T_{\tilde{x}}^{R}(G/K)}\Phi^{(\tilde{x})}%
(\exp_{\tilde{x}}iY)\nu_{t,x}^{\mathrm{c}}(Y)j_{x}^{\mathrm{c}}%
(Y)~dY\nonumber\\
&  =e^{-t\lambda/2}e^{t\left\vert \rho\right\vert ^{2}/2}\int_{T_{\tilde{x}%
}^{R}(G/K)}\Phi^{(\tilde{x})}(\exp_{\tilde{x}}iY)j_{x}^{\mathrm{c}}%
(Y)^{1/2}\frac{e^{-\left\vert Y\right\vert ^{2}/2t}}{(2\pi\tau)^{n/2}}~dY.
\label{atr.int}%
\end{align}

Now, $\Phi^{(\tilde{x})}$ is a $K_{\tilde{x}}$-invariant eigenfunction for
$\Delta$ on $G/K$ with eigenvalue $\lambda.$ We now use an \textquotedblleft
intertwining formula\textquotedblright\ that relates the Laplacian on $G/K$ to
the Euclidean Laplacian on $\mathfrak{p},$ when applied to $K$-invariant
functions. See the proof of Theorem 2 in \cite{HM3} Proposition V.5.1 in
\cite{He3} and the calculations in the complex case on p. 484. The
intertwining formula tells us that the function%
\[
Y\rightarrow\Phi^{(\tilde{x})}(\exp_{\tilde{x}}Y)j_{x}(Y)^{1/2}%
\]
is an eigenfunction for the Euclidean Laplacian on $\mathfrak{p}$ with
eigenvalue $-(\lambda-\left\vert \rho\right\vert ^{2}).$ Therefore, if we
analytically continue and replace $Y$ by $iY$, we conclude that the function%
\[
\Psi(Y):=\Phi^{(\tilde{x})}(\exp_{\tilde{x}}iY)j_{x}^{\mathrm{c}}(Y)^{1/2}%
\]
is an eigenfunction for $\Delta$ on $\mathfrak{p}$ with eigenvalue
$(\lambda-\left\vert \rho\right\vert ^{2}).$ (Recall that $j_{x}^{\mathrm{c}%
}(Y)=j_{x}(iY).$)

We now make use of the following elementary result, which was Lemma 5 of
\cite{HM4}. (We have replaced $2R$ by $R$ in the statement of \cite[Lem.
5]{HM4}.)

\begin{lemma}
\label{euclideanint.lem}Let $\Psi$ be a smooth function on the ball
$B(R_{0},0)$ in $\mathbb{R}^{d}$ satisfying $\Delta\Psi=\sigma\Psi$ for some
constant $\sigma\in\mathbb{R},$ where $\Delta$ is the Euclidean Laplacian. Let
$\beta$ be a bounded, measurable, rotationally invariant function on
$B(R_{0},0).$ Then for all $R<R_{0}$ we have%
\begin{equation}
\int_{\left\vert Y\right\vert \leq R}\Psi(Y)\beta(Y)~dY=\Psi(0)\int%
_{\left\vert Y\right\vert \leq R}e^{\sqrt{\sigma}y_{1}}\beta(Y)~dY.
\label{euclidean.eq}%
\end{equation}
Here $Y=(y_{1},\ldots,y_{d})$ and $\sqrt{\sigma}$ is either of the two square
roots of $\sigma.$
\end{lemma}

Lemma \ref{euclideanint.lem} therefore tells us that the last line in
(\ref{atr.int})\ is equal to%
\[
\Phi^{(\tilde{x})}(\tilde{x})e^{-t\lambda/2}e^{t\left\vert \rho\right\vert
^{2}/2}\int_{\substack{Y\in\mathbb{R}^{d} \\\left\vert Y\right\vert \leq
R}}\exp\left(  \sqrt{\lambda-\left\vert \rho\right\vert ^{2}}~y_{1}\right)
\frac{e^{-\left\vert Y\right\vert ^{2}/2t}}{(2\pi\tau)^{d/2}}~dY.
\]
Since $\Phi^{(\tilde{x})}(\tilde{x})=\Phi(\tilde{x})=\phi(x)$, this
establishes that $A_{t,R}\phi=\alpha_{t,R}(-\Delta)\phi.$ Since $A_{t,R}$ is
known to be bounded and since there exists an orthonormal basis for $L^{2}(X)$
consisting of eigenfunctions of $-\Delta,$ the partial inversion formula follows.
\end{proof}

We now turn to the proof of the global inversion formula, in two versions.
Ultimately, the global inversion formula derives from the partial inversion
formula (Theorem \ref{partial.thm}) together with the observation that
$\lim_{R\rightarrow\infty}\alpha_{t,R}(\lambda)=1$ (see (\ref{alpha.lim})).

\begin{proof}
(Proof of Theorem \ref{global_l2.thm}) For all $R>0,$ we define $A_{t,R}$ to
be $\alpha_{t,R}(-\Delta),$ where $\alpha_{t,R}$ is defined by
(\ref{alpha.int}). The partial inversion formula (Theorem \ref{partial.thm})
tells us that for $R<R_{0},$ $A_{t,R}$ coincides with the operator defined in
(\ref{a_def}). We need to establish, then, that the operator $\alpha
_{t,R}(-\Delta)$ is weakly analytic as a function of $R$ for fixed $t.$ We
choose an orthonormal basis $\{\psi_{n}\}$ for $L^{2}(X)$ consisting of
eigenvectors for $-\Delta,$ with corresponding eigenvalues $\lambda_{n}.$
Since $-\Delta$ has non-negative discrete spectrum, there is some $N$ with
$\lambda_{n}\geq\left\vert \rho\right\vert ^{2}$ for all $n\geq N.$ Given
$f,g\in L^{2}(X),$ we write $f=\sum a_{n}\psi_{n}$ and $g=\sum b_{n}\psi_{n}.$
Then%
\[
\left\langle f,\alpha_{t,R}(-\Delta)g\right\rangle _{L^{2}(X)}=\sum
_{n=1}^{\infty}\overline{a_{n}}b_{n}\alpha_{t,R}(\lambda_{n}).
\]

We use the integral expression (\ref{alpha.int}) for $\alpha_{t,R}$ and we
wish to interchange the integral with the sum over $n.$ To do this we split
off the first $N$ terms and we want to show that Fubini's Theorem applies to
the remaining infinite sum. Note that the exponential in the definition of
$\alpha_{t,R}$ is positive provided that $\lambda_{n}\geq\left\vert
\rho\right\vert ^{2}.$ Thus, if we the integral over $\left\vert Y\right\vert
\leq R$ is bounded by the integral over all of $\mathbb{R}^{d},$ which we have
already remarked is an easy Gaussian integral (see (\ref{alpha.lim})). Thus,
if we put absolute values inside and then interchange the sum and integral, we
get an expression that is bounded by%
\begin{align}
&  \sum_{n=N}^{\infty}\left\vert a_{n}\right\vert \left\vert b_{n}\right\vert
e^{-t\lambda_{n}/2}e^{t\left\vert \rho\right\vert ^{2}/2}\int_{\mathbb{R}^{d}%
}\exp\left(  \sqrt{\lambda_{n}-\left\vert \rho\right\vert ^{2}}~y_{1}\right)
\frac{e^{-\left\vert Y\right\vert ^{2}/2t}}{(2\pi t)^{d/2}}~dY\nonumber\\
&  =\sum_{n=N}^{\infty}\left\vert a_{n}\right\vert \left\vert b_{n}\right\vert
\leq\left\Vert f\right\Vert _{L^{2}(X)}\left\Vert g\right\Vert _{L^{2}%
(X)}<\infty. \label{fubini1}%
\end{align}

We may therefore write%
\begin{align}
&  \left\langle f,\alpha_{t,R}(-\Delta)g\right\rangle _{L^{2}(X)}\nonumber\\
&  =\int_{\substack{Y\in\mathbb{R}^{d} \\\left\vert Y\right\vert \leq
R}}\left[  \sum_{n=1}^{\infty}\overline{a_{n}}b_{n}e^{-t\lambda_{n}%
/2}e^{t\left\vert \rho\right\vert ^{2}/2}\exp\left(  \sqrt{\lambda
_{n}-\left\vert \rho\right\vert ^{2}}~y_{1}\right)  \right]  \frac
{e^{-\left\vert Y\right\vert ^{2}/2t}}{(2\pi t)^{d/2}}~dY. \label{weak}%
\end{align}
It is not hard to show, using Fubini's and Morera's Theorems, that the
expression in square brackets admits an entire holomorphic extension in
$y_{1},$ given by the same formula. From this, it then follows easily that the
integral on the right-hand side of (\ref{weak}) is a real-analytic function of
$R.$ This shows that the operator $\alpha_{t,R}(-\Delta)$ is weakly analytic
as a function of $R,$ which is therefore (in light of Theorem
\ref{partial.thm}) the desired weakly analytic extension of $A_{t,R}. $

Now that we know that the weakly analytic extension of $A_{t,R}$ is given by
$\alpha_{t,R}(-\Delta),$ we need to show that $\alpha_{t,R}(-\Delta)f$ tends
to $f$ in the norm topology of $L^{2}(X),$ for any $f\in L^{2}(X).$ As above,
write $f=\sum a_{n}\psi_{n},$ so that $\alpha_{t,R}(-\Delta)f=\sum
_{n=1}^{\infty}a_{n}\alpha_{t,R}(\lambda_{n})\psi_{n},$ because $\alpha
_{t,R}(-\Delta)$ is bounded. (In both cases, convergence is in $L^{2}(X).$)
Then%
\begin{equation}
\left\Vert f-\alpha_{t,R}(-\Delta)f\right\Vert ^{2}=\sum_{n=1}^{N}%
(1-\alpha_{t,R}(\lambda_{n}))^{2}\left\vert a_{n}\right\vert ^{2}+\sum
_{n=N+1}^{\infty}(1-\alpha_{t,R}(\lambda_{n}))^{2}\left\vert a_{n}\right\vert
^{2}, \label{norm.sum}%
\end{equation}
where again $\lambda_{n}\geq\left\vert \rho\right\vert ^{2}$ for $n>N.$ From
(\ref{alpha.int}), we can see that $\alpha_{t,R}(\lambda)$ is non-negative and
monotone in $R$ for fixed $t$ and $\lambda,$ provided that $\lambda
\geq\left\vert \rho\right\vert ^{2}.$ Since $\lim_{R\rightarrow\infty}%
\alpha_{t,R}(\lambda)=1$ (see (\ref{alpha.lim})) this means that
$0<\alpha_{t,R}(\lambda)\leq1$ for $\lambda\geq\left\vert \rho\right\vert
^{2}.$ Thus Dominated Convergence shows that the second term on the right-hand
side of (\ref{norm.sum}) tends to zero as $R$ tends to infinity. The first
term also tends to zero by (\ref{alpha.lim}), since it is a finite sum. Thus
the left-hand side of (\ref{norm.sum}) tends to zero as $R$ tends to infinity
(with $t$ fixed), establishing the $L^{2}$ form of the global inversion formula.
\end{proof}

We turn now to the pointwise version of the global inversion formula.

\begin{proof}
(Proof of Theorem \ref{global_ptwise.thm}) As in the previous proof, write
$f=\sum_{n=1}^{\infty}a_{n}\psi_{n},$ with convergence in $L^{2}.$ We now
assume that the eigenvectors $\psi_{n}$ are ordered so that the corresponding
eigenvalues are nondecreasing with $n.$ According to Weyl's Law, $\lambda_{n}$
behaves asymptotically like a constant times $n^{2/d}$ as $n$ tends to
infinity, where $d=\dim X,$ as usual. It is also known (e.g., \cite{SZ} and
the references therein) that there is a constant $C,$ depending only on the
choice of $X,$ such that if $\phi$ is an eigenfunction of the Laplacian with
eigenvalue $\lambda$ and normalized to have $L^{2}$ norm 1, then%
\[
\left\Vert \phi\right\Vert _{L^{\infty}}\leq C_{1}\lambda^{(d-1)/4}.
\]
Thus,
\begin{equation}
\left\Vert \psi_{n}\right\Vert _{L^{\infty}}\leq C_{2}(n^{2/d})^{(d-1)/4}%
=C_{2}n^{(d-1)/2d}. \label{sup.est}%
\end{equation}

On the other hand, if $f$ is in the domain of the $l^{\mathrm{th}}$ power of
$\Delta,$ then%
\[
\sum_{n=1}^{\infty}\left\vert a_{n}\right\vert ^{2}\lambda_{n}^{2l}<\infty,
\]
which implies (using Weyl's Law again) that $\left\vert a_{n}\right\vert \leq
Cn^{-2l/d}.$ If, then, $l$ is large enough that%
\[
\varepsilon:=\frac{2l}{d}-\frac{d-1}{2d}-1>0
\]
(which is equivalent to $l>(3d^{2}-d)/4$) we will have%
\begin{align}
\sum_{n=1}^{\infty}\left\vert a_{n}\right\vert \sup\left\vert \psi
_{n}\right\vert  &  \leq C_{3}\sum_{n=1}^{\infty}n^{-2l/d}n^{(d-1)/2d}%
\nonumber\\
&  =C_{3}\sum_{n=1}^{\infty}n^{-(1+\varepsilon)}<\infty. \label{m.series}%
\end{align}
Thus by the Weierstrass $M$-test, the series $\sum a_{n}\psi_{n}$ converges
uniformly as well as in $L^{2}$ to $f.$

Meanwhile, $\alpha_{t,R}(-\Delta)f=\sum_{n}\alpha_{t,R}(\lambda_{n})a_{n}%
\psi_{n},$ with convergence in $L^{2}.$ Since $0\leq\alpha_{t,R}(\lambda
_{n})\leq1$ for $n>N,$ this series also converges uniformly (to $\alpha
_{t,R}(-\Delta)f$). We now plug in the integral formula (\ref{alpha.int}) for
$\alpha_{t,R}$ and we wish to interchange (for each fixed $x$) the sum over
$n$ in $\sum_{n}\alpha_{t,R}(\lambda_{n})a_{n}\psi_{n}$ with the integral in
(\ref{alpha.int}). To do this, we again split off the terms with $n\leq N$ and
argue as in (\ref{fubini1}) for the applicability of Fubini's Theorem in the
remaining terms, substituting the convergence result (\ref{m.series}) for
$\sum\left\vert a_{n}\right\vert \left\vert b_{n}\right\vert \,<\infty.$

We obtain, then,
\begin{align*}
&  \left(  \alpha_{t,R}(-\Delta)f\right)  (x)\\
&  =\sum_{n=1}^{\infty}\alpha_{t,R}(\lambda_{n})a_{n}\psi_{n}(x)\\
&  =\int_{\substack{Y\in\mathbb{R}^{d} \\\left\vert Y\right\vert \leq
R}}\left[  \sum_{n=1}^{\infty}e^{-t\lambda_{n}/2}e^{t\left\vert \rho
\right\vert ^{2}/2}\exp\left(  \sqrt{\lambda_{n}-\left\vert \rho\right\vert
^{2}}~y_{1}\right)  a_{n}\psi_{n}(x)\right]  \frac{e^{-\left\vert Y\right\vert
^{2}/2t}}{(2\pi t)^{d/2}}~dY
\end{align*}
As in the previous proof, the expression in square brackets is an entire
function of $y_{1}$ and the whole integral is a real-analytic function of $R.
$

Meanwhile%
\begin{equation}
f(x)-\left(  \alpha_{t,R}(-\Delta)f\right)  (x)=\sum_{n=1}^{N}(1-\alpha
_{t,R}(\lambda_{n}))a_{n}\psi_{n}(x)+\sum_{n=N+1}^{\infty}(1-\alpha
_{t,R}(\lambda_{n}))a_{n}\psi_{n}(x). \label{pointwise.sum}%
\end{equation}

Because $\alpha_{t,R}(\lambda)\rightarrow1$ as $R\rightarrow\infty,$ the first
term on the right-hand side of (\ref{pointwise.sum}) tends to zero uniformly
as $R$ tends to infinity. Since $0\leq\alpha_{t,R}(\lambda_{n})\leq1$ for
$n>N,$ the absolute value of the second term on the right-hand side of
(\ref{pointwise.sum}) is bounded by%
\begin{equation}
\sum_{n=N+1}^{\infty}(1-\alpha_{t,R}(\lambda_{n}))\left\vert a_{n}\right\vert
\sup\left\vert \psi_{n}\right\vert , \label{bound}%
\end{equation}
independently of $x.$ This expression tends to zero by Dominated Convergence,
in light of (\ref{m.series}). This establishes the desired uniform pointwise
convergence result.
\end{proof}

\section{The isometry formula}

\subsection{Strategy for the isometry formula}

We continue to assume that $G/K$ is a noncompact symmetric space of the
complex type (i.e., $G$ is complex) and that $X=\Gamma\backslash G/K$ is a
compact quotient of $G/K$ of the sort described in Section \ref{setup.sec}.

To obtain the isometry formula for $X,$ we will write (heuristically)%
\begin{align*}
\left\langle f,f\right\rangle _{L^{2}(X)}  &  =\left\langle e^{-t\Delta
/2}e^{t\Delta/2}f,e^{-t\Delta/2}e^{t\Delta/2}f\right\rangle _{L^{2}(X)}\\
&  =\left\langle F,e^{-t\Delta}F\right\rangle _{L^{2}(X)}.
\end{align*}
Note that to compute $e^{-t\Delta}F$ (where $F=e^{t\Delta/2}f$), we want to
apply the backward heat operator for time $2t$, rather than just for time $t.$
Reasoning as in the previous section, we may compute this backward heat
operator by the following integral%
\begin{equation}
e^{-t\Delta}F=\lim_{R\rightarrow\infty}e^{t\left\vert \rho\right\vert ^{2}%
}\int_{T_{x}^{2R}(X)}F(\exp_{x}iY)j_{x}^{\mathrm{c}}(Y)^{1/2}\frac
{e^{-\left\vert Y\right\vert ^{2}/4t}}{(4\pi t)^{d/2}}~dY.
\label{inv.heuristic}%
\end{equation}
Here it is convenient to integrate over a ball of radius $2R$ rather than
radius $R,$ simply to avoid a factor of 2 later on. Heuristically, then, we
should have%
\[
\left\langle f,f\right\rangle _{L^{2}(X)}=\lim_{R\rightarrow\infty
}e^{t\left\vert \rho\right\vert ^{2}}\int_{X}\overline{F(x)}\int_{T_{x}%
^{2R}(X)}F(\exp_{x}iY)j_{x}^{\mathrm{c}}(Y)^{1/2}\frac{e^{-\left\vert
Y\right\vert ^{2}/4t}}{(4\pi t)^{d/2}}~dY~dx.
\]

The crucial next step is a \textquotedblleft holomorphic change of
variable,\textquotedblright\ which will show that (at least for small $R$)%
\begin{align}
&  e^{t\left\vert \rho\right\vert ^{2}}\int_{X}\overline{F(x)}\int_{T_{x}%
^{2R}(X)}F(\exp_{x}iY)j^{\mathrm{c}}(Y)^{1/2}\frac{e^{-\left\vert Y\right\vert
^{2}/4t}}{(4\pi t)^{d/2}}~dY\nonumber\\
&  =e^{t\left\vert \rho\right\vert ^{2}}\int_{X}\int_{T_{x}^{2R}(X)}%
\overline{F(\exp_{x}iY/2)}F(\exp_{x}iY/2)j^{\mathrm{c}}(Y)^{1/2}%
\frac{e^{-\left\vert Y\right\vert ^{2}/4t}}{(4\pi t)^{d/2}}~dY~dx.
\label{change1}%
\end{align}
(Compare Lemma 9 of \cite{H2} and Section 4 of \cite{St1} in the compact
case.) Making the change of variable $Y\rightarrow2Y$ (for sake of
convenience) we obtain a proposal for the form that the isometry theorem
should take:%
\[
\left\langle f,f\right\rangle _{L^{2}(X)}=\lim_{R\rightarrow\infty
}e^{t\left\vert \rho\right\vert ^{2}}\int_{X}\int_{T_{x}^{R}(X)}\left\vert
F(\exp_{x}iY)\right\vert ^{2}j_{x}^{\mathrm{c}}(2Y)\frac{e^{-\left\vert
Y\right\vert ^{2}/t}}{(\pi t)^{d/2}}~dY~dx.
\]
This formula is precisely analogous to what we obtained \cite[Thm. 7]{HM4} for
globally symmetric spaces of the complex type, and is as parallel as possible
to the isometry formula for the dual compact group case (Eq. (\ref{group.isom})).

\subsection{The holomorphic change of variable and the partial isometry
formula}

To proceed rigorously, we consider, for a fixed small value of $R,$ the
integral on the right-hand side of (\ref{inv.heuristic}). We evaluate this
integral by a simple modification of Theorem \ref{partial.thm}. Then we will
establish the holomorphic change of variable in (\ref{change1}), which will
lead to a rigorous version of the isometry formula.

\begin{theorem}
\label{partial_inv2.thm}Let $R_{0}$ be as in Proposition \ref{continue.prop}.
For all $R<R_{0}/2,$ let $B_{t,R}$ be the operator defined by%
\[
B_{t,R}(f)(x)=e^{t\left\vert \rho\right\vert ^{2}}\int_{T_{x}^{2R}(X)}%
F(\exp_{x}iY)j_{x}^{\mathrm{c}}(Y)^{1/2}\frac{e^{-\left\vert Y\right\vert
^{2}/4t}}{(4\pi t)^{d/2}}~dY,
\]
where $F:=e^{t\Delta/2}f.$ Then $B_{t,R}$ is a bounded operator on $L^{2}(X) $
and is given by%
\[
B_{t,R}=\beta_{t,R}(-\Delta),
\]
where $\beta_{t,R}:[0,\infty)\rightarrow\mathbb{R}$ is given by%
\begin{equation}
\beta_{t,R}(\lambda)=e^{-t\lambda/2}e^{t\left\vert \rho\right\vert ^{2}}%
\int_{\substack{Y\in\mathbb{R}^{d} \\\left\vert Y\right\vert \leq2R }%
}\exp\left(  \sqrt{\lambda-\left\vert \rho\right\vert ^{2}}~y_{1}\right)
\frac{e^{-\left\vert Y\right\vert ^{2}/4t}}{(4\pi t)^{d/2}}~dY.
\label{beta.def}%
\end{equation}
Here $\sqrt{\lambda-\left\vert \rho\right\vert ^{2}}$ is either of the two
square roots of $\lambda-\left\vert \rho\right\vert ^{2}.$
\end{theorem}

The proof of this is the same as the proof of Theorem \ref{partial.thm},
except that $t$ is replaced by $2t$ and $R$ by $2R$ in the appropriate places.
Note that if we let $R$ tend to infinity in the definition of $\beta_{t,R},$
we obtain an easily evaluated Gaussian integral, which gives%
\begin{equation}
\lim_{R\rightarrow\infty}\beta_{t,R}(\lambda)=e^{t\lambda/2}. \label{beta.lim}%
\end{equation}
This reflects the idea that $B_{t,R}f$ is an approximation to the backward
heat operator at time $2t,$ applied to the function $F:=e^{t\Delta/2}f.$ Note
that although the right-hand side of (\ref{beta.lim}) is an unbounded function
of $\lambda,$ $\beta_{t,R}$ is a bounded function of $\lambda$ for each fixed
finite value of $R,$ as is easily seen from (\ref{beta.def}).

As in Section \ref{setup.sec}, for each $x$ in $X=\Gamma\backslash G/K,$ we
can pick $\tilde{x}\in G/K$ mapping to $x$ under quotienting by the action of
$\Gamma.$ We can then pick $g\in G$ with $g\cdot x_{0}=\tilde{x},$ where
$x_{0}$ is the identity coset in $G/K.$ Having made these choices, we get an
identification of $T_{x}(X)$ with $T_{x_{0}}(G/K)=\mathfrak{p}.$ Two
identifications of this sort differ only by the action of $K$ on
$\mathfrak{p}.$ This means that if we have some $K$-invariant function on
$\mathfrak{p},$ we can transfer this function in an unambiguous way to each
$T_{x}(X).$

\begin{theorem}
[Holomorphic Change of Variable]\label{holo_change.thm}Let $F_{1}$ and $F_{2}$
be holomorphic functions on $T^{R_{0}}(X)$ for some $R_{0}>0.$ Let $\alpha$ be
a bounded, measurable, $K$-invariant function on $\mathfrak{p}^{R_{0}}.$ Then
for all $R<R_{0}/2$ we have%
\begin{align}
&  \int_{X}\overline{F_{1}(x)}\int_{T_{x}^{2R}(X)}F_{2}(\exp_{x}%
iY)\alpha(Y)~dY~dx\nonumber\\
&  =\int_{X}\int_{T_{x}^{2R}(X)}\overline{F_{1}(\exp_{x}(iY/2))}F_{2}(\exp
_{x}(iY/2))\alpha(Y)~dY~dx. \label{change.eq}%
\end{align}

\end{theorem}

Note that although the right-hand side of (\ref{change.eq}) is defined for all
$R<R_{0},$ the left-hand side is defined only for $R<R_{0}/2.$ Once this
result is established, we will apply Theorem \ref{partial_inv2.thm} and the
holomorphic change of variable with $\alpha$ given by
\begin{equation}
\alpha(Y)=\nu_{2t}^{\mathrm{c}}(Y)j^{\mathrm{c}}(Y)=e^{t\left\vert
\rho\right\vert ^{2}}j^{\mathrm{c}}(Y)^{1/2}\frac{e^{-\left\vert Y\right\vert
^{2}/4t}}{(4\pi t)^{d/2}}. \label{alpha_form}%
\end{equation}
After making the change of variable $Y\rightarrow2Y$, for convenience, we will
obtain the following result, which is the main result of this subsection.

\begin{theorem}
[Partial Isometry Formula]\label{partial_isom.thm}Given $f_{1},f_{2}\in
L^{2}(X),$ let $F_{1}=e^{t\Delta/2}f_{1}$ and $F_{2}=e^{t\Delta/2}f_{2}.$ Let
$R_{0}$ be as in Proposition \ref{continue.prop}. Then for all $R<R_{0},$%
\begin{align}
&  e^{t\left\vert \rho\right\vert ^{2}}\int_{X}\int_{T_{x}^{R}(X)}%
\overline{F_{1}(\exp_{x}(iY))}F_{2}(\exp_{x}(iY))j^{\mathrm{c}}(2Y)^{1/2}%
\frac{e^{-\left\vert Y\right\vert ^{2}/t}}{(\pi t)^{d/2}}~dY~dx\nonumber\\
&  =\left\langle f_{1},e^{t\Delta/2}\beta_{t,R}(-\Delta)f_{2}\right\rangle
_{L^{2}(X)}, \label{partial_isom.eq}%
\end{align}
where $\beta_{t,R}$ is the function defined in (\ref{beta.def}).
\end{theorem}

The analogous result on the globally symmetric space $G/K$ ($G$ complex) was
obtained in \cite{HM4}; see Theorem 6 and Equations (38) and (39). We will
prove (\ref{partial_isom.eq}) directly from the holomorphic change of variable
for $R<R_{0}/2$ and then extend the result to $R<R_{0}$ by analytic continuation.

From the formula (\ref{jc_half}), we see that $j^{\mathrm{c}}(2Y)$ is positive
for all sufficiently small $Y.$ (Actually, it is not hard to show that
$j^{\mathrm{c}}(2Y)$ is positive on the maximal domain in $T(X)$ on which the
adapted complex structure exists, but we do not require this result.) This
means that the left-hand side of (\ref{partial_isom.eq}) is strictly positive
whenever $f_{1}=f_{2}=f,$ where $f$ is nonzero. If follows that $\beta
_{t,R}(-\Delta)$ is a strictly positive operator, for all sufficiently small
$R,$ something that can also be obtained from the formula for $\beta_{t,R}.$
However, because the spectrum of $-\Delta$ contains points in the interval
$[0,\left\vert \rho\right\vert ^{2}),$ it is not clear whether positivity
holds for all $R.$

Note that the operator $B_{t,R}$ in Theorem \ref{partial_inv2.thm}
incorporates the $e^{t\Delta/2}$ applied to $f_{1},$ but not the
$e^{t\Delta/2}$ applied to $f_{2}.$ This, along with the self-adjointness of
$e^{t\Delta/2},$ accounts for the expression on the right-hand side of
(\ref{partial_isom.eq}).

To understand what is going on in Theorem \ref{holo_change.thm}, it is helpful
to consider the following prototype calculation on the (Euclidean) symmetric
space $S^{1}=\mathbb{R}/\mathbb{Z}.$ Then $T(\mathbb{R}/\mathbb{Z}) $, with
the adapted complex structure, is identified with $\mathbb{C}/\mathbb{Z}$ in
such a way that $\exp_{x}(iy)=x+iy.$ Suppose $F_{1}$ and $F_{2}$ are
holomorphic functions on a strip in $\mathbb{R}/\mathbb{Z}.$ Let $\tilde
{F}_{1}$ be the holomorphic function whose restriction to $\mathbb{R}%
/\mathbb{Z}$ is $\overline{F_{1}}$; equivalently, $\tilde{F}_{1}%
(x+iy)=\overline{F_{1}(x-iy)}.$ Then using Fubini and a change of variable we
have%
\[
\int_{\mathbb{R}/\mathbb{Z}}\overline{F_{1}(x)}\int_{-2R}^{2R}F_{2}%
(x+iy)\alpha(y)~dy~dx=\int_{-2R}^{2R}\int_{\mathbb{R}/\mathbb{Z}}\tilde{F}%
_{1}(x-a)F_{2}(x-a+iy)~dx~\alpha(y)~dy
\]
for any $a\in\mathbb{R}.$ Since $\tilde{F}_{1}$ and $F_{2}$ are holomorphic,
this equality remains valid for $a$ in a strip in $\mathbb{C}.$ Taking
$a=iy/2$ and using Fubini again gives%
\[
\int_{\mathbb{R}/\mathbb{Z}}\overline{F_{1}(x)}\int_{-2R}^{2R}F_{2}%
(x+iy)\alpha(y)~dy~dx=\int_{\mathbb{R}/\mathbb{Z}}\int_{-2R}^{2R}%
\overline{F_{1}(x+iy/2)}F_{2}(x+iy/2)\alpha(y)~dy~dx.
\]
This is just the analog of Theorem \ref{holo_change.thm} for $\mathbb{R}%
/\mathbb{Z}$ and the method of proof (the \textquotedblleft change of
variable\textquotedblright\ $x\rightarrow x-iy/2$) motivates the terminology
\textquotedblleft holomorphic change of variable.\textquotedblright

The remainder of this subsection is devoted to the proof of Theorem
\ref{holo_change.thm}. The first step is to express the two sides of
(\ref{change.eq}) in terms of \textquotedblleft double orbital
integrals,\textquotedblright\ with the integration being over $\Gamma
\backslash G$ with respect to the natural right-$G$-invariant measure. Then a
\textquotedblleft holomorphic change of variable\textquotedblright\ in the
double orbital integrals, similar to that in the previous paragraph, will show
that the integrals are equal.\bigskip

Let $\pi:G/K\rightarrow X=\Gamma\backslash G/K$ be the quotient map and let
$x_{0}$ denote the identity coset in $G/K.$ For each $x\in X,$ we choose
$g_{x}\in G$ so that $\pi(g_{x}\cdot x_{0})=x.$ We arrange for the $g_{x}$'s
to depend measurably on $x,$ and we will make one other technical restriction
on the choice of $g_{x}$ later. We then identify each $T_{x}(X)$ with
$T_{g_{x}\cdot x_{0}}(G/K)$ by means of $\pi_{\ast}^{-1}$ and then with
$\mathfrak{p}=T_{x_{0}}(G/K)$ by the action of $g_{x}^{-1}.$ This
identification of $T_{x}(X)$ with $\mathfrak{p}$ is of the same sort as we
have been using all along in this paper, but we now have one particular such
identification for each $x.$

We have now measurably identified $T(X)$ with $X\times\mathfrak{p}.$ We use
this identification on both sides of the desired equality (\ref{change.eq}),
along with generalized polar coordinates on $\mathfrak{p}$ with respect to the
adjoint action of $K.$ Recall (from (\ref{exp.id})) that $\exp_{x}(iY)$ is
simply another way of writing the point $(x,Y)\in T(X).$ Let us now switch
back to the $(x,Y)$ notation. Then, using our identification of $T(X)$ with
$X\times\mathfrak{p}$ and generalized polar coordinates on $\mathfrak{p},$ the
desired equality (\ref{change.eq}) is equivalent to%
\begin{align}
&  \int_{\mathfrak{a}_{2R}^{+}}\int_{X}\int_{K}\overline{F_{1}(x)}%
F_{2}((x,\mathrm{Ad}_{k}Y))~dk~dx~\alpha(Y)\mu(Y)~dY\nonumber\\
&  =\int_{\mathfrak{a}_{2R}^{+}}\int_{X}\int_{K}\overline{F_{1}((x,\mathrm{Ad}%
_{k}Y/2))}F_{2}((x,\mathrm{Ad}_{k}Y/2))~dk~dx~\alpha(Y)\mu(Y)~dY.
\label{change2}%
\end{align}
Here, $dx$ denotes the Riemannian volume measure on $X$, $\mu$ is the density
that appears in the generalized polar coordinates formula (e.g., \cite[Thm.
I.5.17]{He2}), and $\mathfrak{a}_{2R}^{+}$ is the set of vectors in
$\mathfrak{a}^{+}$ with norm less than $2R.$

Clearly, for (\ref{change2}) to hold, it is sufficient to verify that
\begin{align}
&  \int_{X}\int_{K}\overline{F_{1}(x)}F_{2}((x,\mathrm{Ad}_{k}%
Y))~dk~dx\nonumber\\
&  =\int_{X}\int_{K}\overline{F_{1}((x,\mathrm{Ad}_{k}Y/2))}F_{2}%
((x,\mathrm{Ad}_{k}Y/2))~dk~dx \label{change3}%
\end{align}
for all $Y\in\mathfrak{a}_{2R}^{+}.$

Our goal now is to show that both integrals in (\ref{change3}) can be written
as \textquotedblleft double orbital integrals.\textquotedblright\ Let
$\tilde{F}_{1}$ be the function on $T^{R_{0}}(X)$ given by $\tilde{F}%
_{1}(x,Y)=\overline{F_{1}(x,-Y)}$ (or, by (\ref{exp.id}), $\tilde{F}_{1}%
(\exp_{x}iY)=\overline{F_{1}(\exp_{x}(-iY))}$). Since the map
$(x,Y)\rightarrow(x,-Y)$ is antiholomorphic \cite[p. 568]{GStenz1}, $\tilde
{F}_{1}$ is holomorphic. Now, it is known that $T^{R}(G/K)$ has its own
adapted complex structure for all sufficiently small $R$ (see \cite{HM4} and
the references therein). Furthermore, the map $\pi_{\ast}:T^{R}%
(G/K)\rightarrow T^{R}(\Gamma\backslash G/K)$ is easily seen to be
holomorphic. Thus, we can define holomorphic functions $\Phi_{1}$ and
$\Phi_{2}$ on $T^{R}(G/K)$ by%
\begin{align*}
\Phi_{1}  &  =\tilde{F}_{1}\circ\pi_{\ast}\\
\Phi_{2}  &  =F_{2}\circ\pi_{\ast}.
\end{align*}
By construction, these functions satisfy $\Phi_{j}(\gamma\cdot a)$ for all
$\gamma\in\Gamma$ and $a\in T^{R}(G/K),$ where $\gamma\cdot a$ refers to the
action of $\Gamma$ on $T^{R}(G/K)$ induced from the action of $\Gamma$ on
$G/K.$

We now consider \textquotedblleft double orbital integral,\textquotedblright%
\ namely, integrals of the form%
\begin{equation}
\int_{\Gamma\backslash G}\Phi_{1}(g\cdot a)\Phi_{2}(g\cdot b)~dg,
\label{double}%
\end{equation}
where $dg$ is an appropriately normalized version of the right-$G$-invariant
measure on $\Gamma\backslash G$ and $a$ and $b$ are fixed points in
$T^{R}(G/K).$ Observe that although $g\cdot a$ and $g\cdot b$ are defined for
$g\in G$ (not $\Gamma\backslash G$), the invariance property of the $\Phi_{j}%
$'s means that the integrand in (\ref{double}) descends to a function on
$\Gamma\backslash G.$

\begin{lemma}
\label{orbit.lemma}The left-hand side of (\ref{change3}) is an integral of the
form (\ref{double}) with $a=x_{0}$ and $b=(x_{0},Y).$ The right-hand side of
(\ref{change3}) is an integral of the form (\ref{double}) with $a=(x_{0}%
,-Y/2)$ and $b=(x_{0},Y/2).$
\end{lemma}

\begin{proof}
Let $E$ denote the following set in $G$:%
\[
E=\left\{  \left.  g_{x}k\in G\right\vert x\in X,~k\in K\right\}  .
\]
We assume that the mapping $x\rightarrow g_{x}$ is chosen in such a way that
$E$ is a measurable subset of $G.$

Now consider the map from $X\times K$ to $\Gamma\backslash G$ given by
$(x,k)\rightarrow\Gamma g_{x}k.$ This map is measurable because the map
$x\rightarrow g_{x}$ was chosen to be measurable. Given $\Gamma g\in
\Gamma\backslash G,$ let $x$ be the point $\Gamma gK\in\Gamma\backslash G/K=X$
and then consider $g_{x}\in G,$ which has the property that $\Gamma g_{x}K=x.
$ Then there exists $k\in K$ with $\Gamma g_{x}k=\Gamma g.$ To see that $k$ is
unique, suppose $\gamma_{1}g_{x}k_{1}=\gamma_{2}g_{x}k_{2}$ for some
$\gamma_{1},\gamma_{2}\in\Gamma$ and $k_{1},k_{2}\in K.$ Then because $\Gamma$
acts freely on $G/K,$ we must have $\gamma_{1}=\gamma_{2} $ and therefore
$k_{1}=k_{2}.$

This argument shows that the map $(x,k)\rightarrow\Gamma g_{x}k$ is a
bijection of $X\times K$ onto $\Gamma\backslash G.$ We may therefore identify
$\Gamma\backslash G$ with the set $E\subset G$ defined above. Consider on
$X\times K$ the product of the Riemannian volume measure on $X$ and the
normalized Haar measure on $K.$ The push-forward of this measure to $E$ is
easily seen to be the restriction to $E$ of a (bi-invariant) Haar measure on
$G.$ (Specifically, arguing as in the proof of Proposition 4 in \cite{HM4},
the push-forward measure is the restriction to $E$ of a \textit{left} Haar
measure on $G,$ which is also right invariant since $G$ is unimodular.) It is
then easy to see that if we identify $E$ with $\Gamma\backslash G,$ the
resulting measure on $\Gamma\backslash G$ is invariant under the right action
of $G.$

All of this is to say that if we write points in $\Gamma\backslash G$ as
$\Gamma g_{x}k,$ with $x\in X$ and $k\in K,$ then the measure $dg$ on
$\Gamma\backslash G$ decomposes as $dx~dk.$ Meanwhile, in light of the
identifications we are making,%
\[
\pi_{\ast}[(g_{x}k)\cdot(x_{0},Z)]=(x,\mathrm{Ad}_{k}Z),\quad Z\in
\mathfrak{p}.
\]
The lemma then follows by plugging in $Z=0,$ $Z=Y,$ $Z=-Y/2$, and $Z=Y/2.$
\end{proof}

\begin{lemma}
\label{holo.lemma}

\begin{enumerate}
\item Integrals of the form (\ref{double}) satisfy%
\[
\int_{\Gamma\backslash G}\Phi_{1}(g\cdot(h\cdot a))\Phi_{2}(g\cdot(h\cdot
b))~dg=\int_{\Gamma\backslash G}\Phi_{1}(g\cdot a)\Phi_{2}(g\cdot b)~dg
\]
for all $h\in G$ and $a,b\in T^{R}(G/K).$

\item The integral in (\ref{double}) depends holomorphically on $a$ and $b$
(with $F_{1}$ and $F_{2}$ and hence $\Phi_{1}$ and $\Phi_{2}$ fixed).
\end{enumerate}
\end{lemma}

\begin{proof}
The first point follows from the associativity property of the action and the
right-$G$-invariance of the measure on $\Gamma\backslash G.$ The second point
follows from Morera's Theorem.
\end{proof}

We are now ready to give the proof of the holomorphic change of variable.

\begin{proof}
(Proof of Theorem \ref{holo_change.thm}) It suffices to prove (\ref{change3}),
which, by Lemma \ref{orbit.lemma} amounts to showing that the (double) orbital
integral with $a=x_{0}$ and $b=(x_{0},Y)$ is the same as the orbital integral
with $a=(x_{0},-Y/2)$ and $b=(x_{0},Y/2).$ The idea is to use invariance of
the orbital integral under $(a,b)\rightarrow(e^{tY}\cdot a,e^{tY}\cdot b)$ and
then plug in $t=-i/2,$ using Lemma \ref{holo.lemma}.

There is nothing to prove in (\ref{change3}) if $Y=0.$ If $Y\neq0,$ let
$r=\left\vert Y\right\vert $ (with $r<R_{0}$) and let $X=Y/r$ be the
associated unit vector. Let $\gamma$ be the corresponding unit-speed geodesic
in $G/K,$ namely, $\gamma(t)=e^{tX}\cdot x_{0}.$ Then consider the map
$\tau:S_{R}\rightarrow T(G/K)$ given by
\[
\tau(u+iv)=(\gamma(u),v\dot{\gamma}(u)),
\]
where $S_{R}\subset\mathbb{C}$ is the strip defined in (\ref{sr}). According
to the definition of the adapted complex structure on $T(G/K),$ the map $\tau$
is holomorphic. Note that%
\[
\frac{d}{dt}e^{tX}\cdot x_{0}=\frac{d}{d\varepsilon}\left.  e^{(t+\varepsilon
)X}\cdot x_{0}\right\vert _{\varepsilon=0}=\frac{d}{d\varepsilon}\left.
e^{tX}\cdot(e^{\varepsilon X}\cdot x_{0})\right\vert _{\varepsilon=0}%
=e^{tX}\cdot(x_{0},X),
\]
where in the last expression we have the induced action of $e^{tX}$ on
$T(G/K).$ It follows that
\[
\tau(u+iv)=e^{uX}\cdot(x_{0},vX).
\]

By the first part of Lemma \ref{holo.lemma}, the orbital integral with
$a=x_{0}$ and $b=(x_{0},rX)$ is the same as the orbital integral with
$a=e^{tX}\cdot x$ and $b=e^{tX}\cdot(x_{0},rX),$ for all $t\in\mathbb{R}.$
That is, the orbital integral associated to $a=\tau(t)$ and $b=\tau(t+ir)$ is
independent of $t$ for $t\in\mathbb{R}.$ Since $\tau$ is holomorphic, the
second part of Lemma \ref{holo.lemma} tells us the same result for all
$t\in\mathbb{C}$ such that both $t$ and $t+ir$ belong to the strip $S_{R}.$
Equating the orbital integral with $(a,b)=(\tau(0),\tau(ir))$ (i.e., $t=0$) to
the one with $(a,b)=(\tau(-ir/2),\tau(ir/2))$ (i.e., $t=-ir/2$) gives the
desired result.
\end{proof}

\begin{proof}
(Proof of Theorem \ref{partial_isom.thm}.) For $R<R_{0}/2,$ this result
follows from applying the holomorphic change of variable with $\alpha$ as in
(\ref{alpha_form}) (and then making the cosmetic change of variable
$Y\rightarrow2Y$). The result will hold for all $R<R_{0}$, provided we can
show that both sides of (\ref{partial_isom.eq}) are real-analytic in $R.$ The
analyticity of the right-hand side of (\ref{partial_isom.eq}) is established
as part of the proof of Theorem \ref{global_isom.thm} in Section
\ref{global_isom.sec}. The analyticity of the left-hand side is established in
the following lemma.
\end{proof}

\begin{lemma}
\label{gf.analytic}Let $R_{1}$ be any positive real number such that the
adapted complex structure exists on $T^{R_{1}}(X).$ Let $F$ be a holomorphic
function on $T^{R_{1}}(X)$ and let $\alpha$ be a real-analytic, Ad-$K$%
-invariant function on $\mathfrak{p}^{2R_{1}}.$ Then the function $G_{F}$
defined by%
\[
G_{F}(R)=\int_{X}\int_{T_{x}^{2R}(X)}\overline{F_{1}(\exp_{x}(iY/2))}%
F_{2}(\exp_{x}(iY/2))\alpha(Y)~dY~dx
\]
is real-analytic on the interval $(0,R_{1}).$
\end{lemma}

\begin{proof}
We measurably identify $T(X)$ with $X\times\mathfrak{p}$ and then decompose
$G_{F}(R)$ as in the right-hand side of (\ref{change2}). We let $\Phi
_{1}=\tilde{F}_{1}\circ\pi_{\ast}$ and we let $\Phi_{2}=F_{2}\circ\pi_{\ast},$
as before. In light of Lemma \ref{holo.lemma}, we can then write our function
as%
\begin{align*}
G_{F}(R)  &  =\int_{\mathfrak{a}_{2R}^{+}}\int_{\Gamma\backslash G}\Phi
_{1}(g\cdot(x_{0},Y/2))\Phi_{2}(g\cdot(x_{0},-Y/2))~dg~\alpha(Y)\mu(Y)~dY\\
&  =\int_{\mathfrak{a}_{1}^{+}}\int_{\Gamma\backslash G}\Phi_{1}(g\cdot
(x_{0},RZ))\Phi_{2}(g\cdot(x_{0},-RZ))~dg~(2R)^{d}~\alpha(2RZ)\mu(2RZ)~dZ.
\end{align*}
Let $\tau_{Z}$ be the map from the strip $S_{R_{1}}\subset\mathbb{C}$ into
$T^{R_{1}}(G/K)$ given by $\tau_{Z}(u+iv)=e^{uZ}\cdot(x_{0},vZ).$ As in the
proof of Theorem \ref{holo_change.thm}, this map is holomorphic. (There it was
assumed that $Z$ was a unit vector, but by scaling the statement is true for
any $Z$ with $\left\vert Z\right\vert <1.$) Note that $(x_{0},RZ)=\tau
_{Z}(iR)$ and $(x_{0},-RZ)=\tau_{Z}(-iR).$ Then given any $R\in(0,R_{1}),$ we
claim that $G_{F}$ has a holomorphic extension to a small neighborhood of $R$
in $\mathbb{C},$ given by%
\begin{equation}
G_{F}(S)=\int_{\mathfrak{a}_{1}^{+}}\int_{\Gamma\backslash G}\Phi_{1}%
(g\cdot\tau_{Z}(iS))\Phi_{2}(g\cdot\tau_{Z}(-iS))~dg~(2S)^{d}~\alpha
(2SZ)\mu(2SZ)~dZ. \label{holo.ext}%
\end{equation}
This follows from Point 2 of Lemma \ref{holo.lemma} and the analyticity of
$\alpha$ (assumed) and $\mu$ (it is a polynomial).

To be a bit more precise, $\alpha$ is assume to be real-analytic on
$\mathfrak{p}^{2R_{1}}.$ Thus, for each $R<R_{1},$ $\alpha$ has a holomorphic
extension to a neighborhood $U$ in $\mathfrak{p}_{\mathbb{C}}$ of the closed
ball of radius $2R$ in $\mathfrak{p}.$ Then for all $S$ in a neighborhood of
$R,$ $2SZ$ will belong to $U$ for all $Z\in\mathfrak{p}$ with $\left\vert
Z\right\vert \leq1.$ For $S$ in a slightly smaller neighborhood of $R,$
$\alpha(2SZ)$ will be bounded uniformly in $Z$ with $\left\vert Z\right\vert
\leq1.$ The integrand in (\ref{holo.ext}) is thus holomorphic in $S$ and
bounded uniformly in $g$ and $Z,$ from which it follows that the integral is
holomorphic in $S.$

The existence of this holomorphic extension establishes the real-analyticity
of $G_{F}.$
\end{proof}

\subsection{The global isometry formula\label{global_isom.sec}}

\begin{theorem}
[Global Isometry Formula]\label{global_isom.thm}Given $f\in L^{2}(X)$ and let
$F=e^{t\Delta/2}f.$ Let $R_{0}$ be as in Proposition \ref{continue.prop}. Then
the quantity
\[
G_{F}(R):=e^{t\left\vert \rho\right\vert ^{2}}\int_{X}\int_{T_{x}^{R}%
(X)}\left\vert F(\exp_{x}(iY))\right\vert ^{2}j_{x}^{\mathrm{c}}%
(2Y)^{1/2}\frac{e^{-\left\vert Y\right\vert ^{2}/t}}{(\pi t)^{d/2}}~dY~dx,
\]
initially defined for $R\in(0,R_{0}),$ has a real-analytic extension to
$R\in(0,\infty).$ Furthermore, this real-analytic extension, also denoted
$G_{F},$ satisfies%
\[
\lim_{R\rightarrow\infty}G_{R}(R)=\left\Vert f\right\Vert _{L^{2}(X)}^{2}.
\]
Thus, we may write, informally,%
\begin{align*}
&  \left\Vert f\right\Vert _{L^{2}(X)}^{2}\\
&  =\text{\textquotedblleft}\lim_{R\rightarrow\infty}{}%
\text{\textquotedblright}e^{t\left\vert \rho\right\vert ^{2}}\int_{X}%
\int_{T_{x}^{R}(X)}\left\vert F(\exp_{x}(iY))\right\vert ^{2}j_{x}%
^{\mathrm{c}}(2Y)^{1/2}\frac{e^{-\left\vert Y\right\vert ^{2}/t}}{(\pi
t)^{d/2}}~dY~dx.
\end{align*}

\end{theorem}

This result, as for the inversion formulas, is obtained from the corresponding
result in the dual compact group case (see (\ref{group.isom})) by
\textquotedblleft dualizing\textquotedblright\ and inserting an analytic
continuation with respect to $R.$

In the corresponding theorem for $G/K$ ($G$ complex) in \cite{HM4}, the proof
actually shows that $G_{F}(R)$ is positive and strictly increasing as a
function of $R.$ In the case of the compact quotient $X=\Gamma\backslash G/K$,
however, the presence of spectrum for $-\Delta$ in the interval $[0,\left\vert
\rho\right\vert ^{2})$ means that $G_{F}(R)$ is not necessarily monotone in
$R$, once $R>R_{0}.$ (See (\ref{partial_isom.eq3}) below.)

\begin{proof}
Putting $f_{1}=f_{2}=f$ in the partial isometry formula (Theorem
\ref{partial_isom.thm}), we have (with $F=e^{t\Delta/2}f$ as usual)%
\begin{equation}
G_{F}(R)=\left\langle f,e^{t\Delta/2}\beta_{t,R}(-\Delta)f\right\rangle
_{L^{2}(X)}. \label{partial_isom.eq2}%
\end{equation}
The point is now that the definition of $\beta_{t,R}$ makes sense for any
$R>0.$ Thus, the right-hand side of (\ref{partial_isom.eq2}) makes sense for
all $R>0,$ even though the left-hand side is defined only for small $R.$
Equation (\ref{beta.lim}) then suggests that the right-hand side of
(\ref{partial_isom.eq2}) should tend to $\left\langle f,f\right\rangle
_{L^{2}(X)} $ as $R$ tends to infinity.

To proceed rigorously, we choose an orthonormal basis $\{\psi_{n}\}$ for
$L^{2}(X)$ consisting of eigenvectors of $-\Delta$ with eigenvalues
$\lambda_{n}.$ Then if $f=\sum a_{n}\psi_{n},$ we have, for any $R>0,$%
\begin{align}
&  \left\langle f,e^{t\Delta/2}\beta_{t,R}(-\Delta)f\right\rangle _{L^{2}%
(X)}\nonumber\\
&  =\sum_{n=1}^{\infty}\left\vert a_{n}\right\vert ^{2}e^{-t\lambda_{n}%
t/2}\beta_{t,R}(\lambda_{n})\nonumber\\
&  =\sum_{n=1}^{\infty}\left\vert a_{n}\right\vert ^{2}e^{-t\lambda_{n}%
}e^{t\left\vert \rho\right\vert ^{2}}\int_{\left\vert Y\right\vert \leq R}%
\exp\left(  \sqrt{\lambda_{n}-\left\vert \rho\right\vert ^{2}}~y_{1}\right)
\frac{e^{-\left\vert Y\right\vert ^{2}/4t}}{(4\pi t)^{d/2}}~dY.
\label{partial_isom.eq3}%
\end{align}
The same argument as in the proof of Theorem \ref{global_l2.thm} shows that
Fubini's Theorem applies, so that we obtain%
\begin{align}
&  \left\langle f,e^{t\Delta/2}\beta_{t,R}(-\Delta)f\right\rangle _{L^{2}%
(X)}\nonumber\\
&  =e^{ct}\int_{\left\vert Y\right\vert \leq R}\left[  \sum_{n=1}^{\infty
}\left\vert a_{n}\right\vert ^{2}e^{-t\lambda_{n}}\exp\left(  \sqrt
{\lambda_{n}-\left\vert \rho\right\vert ^{2}}~y_{1}\right)  \right]
\frac{e^{-\left\vert Y\right\vert ^{2}/4t}}{(4\pi t)^{d/2}}~dY.
\label{isom.fubini}%
\end{align}
Arguing, again, as in the proof of Theorem \ref{global_l2.thm}, we can see
that the expression in square brackets has an entire holomorphic extension to
$y_{1}\in\mathbb{C}$ (given by the same expression) and that the whole
right-hand side of (\ref{isom.fubini}) is real-analytic as a function of $R,$
for all $R\in(0,\infty).$

We have established, then, that the right-hand side of (\ref{partial_isom.eq2}%
) is a real-analytic function of $R$ for $R\in(0,\infty). $ This, along with
Lemma \ref{gf.analytic}, shows that (\ref{partial_isom.eq2}) holds for all
$R<R_{0}.$ (This was initially established, using the holomorphic change of
variable, only for $R<R_{0}/2.$) Thus, the right-hand side of
(\ref{partial_isom.eq2}) is the desired real-analytic extension of $G_{F}.$ To
evaluate the limit as $R$ tends to infinity of this expression, we use our
orthonormal basis $\{\psi_{n}\}$ and we find $N$ so that $\lambda_{n}%
\geq\left\vert \rho\right\vert ^{2}$ for $n>N.$ Then%
\begin{align}
&  \left\langle f,e^{t\Delta/2}\beta_{t,R}(-\Delta)f\right\rangle _{L^{2}%
(X)}\nonumber\\
&  =\sum_{n=1}^{N}\left\vert a_{n}\right\vert ^{2}e^{-t\lambda_{n}/2}%
\beta_{t,R}(\lambda_{n})+\sum_{n=N+1}^{\infty}\left\vert a_{n}\right\vert
^{2}e^{-t\lambda_{n}/2}\beta_{t,R}(\lambda_{n}). \label{isom.lim}%
\end{align}
Note that (by \ref{beta.lim}) we have%
\[
\lim_{R\rightarrow\infty}e^{-t\lambda/2}\beta_{t,R}(\lambda)=1
\]
for all $\lambda\geq0.$ For $n>N,$ $\beta_{t,R}(\lambda_{n})$ is positive and
increasing with $R$; thus, by Monotone Convergence, we can put the limit as
$R$ tends to infinity inside the infinite sum in (\ref{isom.lim}). This shows
that%
\[
\lim_{R\rightarrow\infty}\left\langle f,e^{t\Delta/2}\beta_{t,R}%
(-\Delta)f\right\rangle _{L^{2}(X)}=\sum_{n=1}^{\infty}\left\vert
a_{n}\right\vert ^{2}=\left\Vert f\right\Vert _{L^{2}(X)}^{2}.
\]
This, in light of (\ref{partial_isom.eq2}), is what we want to prove.
\end{proof}

\subsection{The surjectivity theorem}

We now show, roughly, that if $F$ is any holomorphic function for which the
right-hand side of the global isometry formula makes sense and is finite, then
$F$ is the analytic continuation of a function of the form $e^{t\Delta/2}f,$
with $f\in L^{2}(X).$

\begin{theorem}
\label{surj.thm}Suppose $F$ is a holomorphic function on $T^{R_{1}}(X),$ for
some $R_{1}>0$ such that the adapted complex structure exists on $T^{R_{1}%
}(X).$ Let $G_{F}$ be the function defined by%
\[
G_{F}(R)=e^{t\left\vert \rho\right\vert ^{2}}\int_{X}\int_{T_{x}^{R}%
(X)}\left\vert F(\exp_{x}(iY))\right\vert ^{2}j_{x}^{\mathrm{c}}%
(2Y)^{1/2}\frac{e^{-\left\vert Y\right\vert ^{2}/t}}{(\pi t)^{d/2}}~dY~dx,
\]
for $R<R_{1}.$ Suppose that $G_{F}$ has a real-analytic extension (also
denoted $G_{F}$) from $(0,R_{1})$ to $(0,\infty)$ and that%
\[
\lim_{R\rightarrow\infty}G_{F}(R)
\]
exists and is finite. Then there exists a unique $f\in L^{2}(X)$ for which
$\left.  F\right\vert _{X}=e^{t\Delta/2}f.$
\end{theorem}

Choose $R_{0}$ so as in Proposition \ref{continue.prop} and then choose
$R_{2}\leq R_{0}$ so that the function $j^{\mathrm{c}}(2Y)$ is positive on
$T^{R_{2}}(X).$ For any $R<R_{2},$ let $\mathcal{H}L^{2}(T^{R}(X))_{t}$ denote
the space of holomorphic functions $F$ on $T^{R}(X)$ for which%
\[
e^{t\left\vert \rho\right\vert ^{2}}\int_{X}\int_{T_{x}^{R}(X)}\left\vert
F(\exp_{x}(iY))\right\vert ^{2}j_{x}^{\mathrm{c}}(2Y)^{1/2}\frac
{e^{-\left\vert Y\right\vert ^{2}/t}}{(\pi t)^{d/2}}~dY~dx<\infty,
\]
with the obvious associated inner product. A standard argument shows that
$\mathcal{H}L^{2}(T^{R}(X))_{t}$ is a closed subspace of the associated $L^{2}
$ space, and hence a Hilbert space. Note that if $\psi\in L^{2}(X)$ is an
eigenvector for $-\Delta$ with eigenvalue $\lambda,$ then $\psi=e^{t\Delta
/2}(e^{\lambda t/2}\psi),$ so that by Proposition \ref{continue.prop}, $\psi$
has an analytic continuation to $T^{R_{0}}(X).$ This analytic continuation is
bounded on each $T^{R}(X)$ for $R<R_{0}$ and hence belongs to $\mathcal{H}%
L^{2}(T^{R}(X))_{t}.$

\begin{lemma}
Let $\{\psi_{n}\}$ be an orthonormal basis for $L^{2}(X)$ consisting of
eigenvectors for $-\Delta.$ Let $\psi_{n}$ also denote the analytic
continuation of $\psi_{n}$ to $T^{R}(X).$ Then the $\psi_{n}$'s form an
orthogonal basis for $\mathcal{H}L^{2}(T^{R}(X))_{t}$, for all $R<R_{2}$ and
$t>0.$
\end{lemma}

\begin{proof}
We fix one particular $R<R_{2}$ and $t>0$, and we abbreviate $\mathcal{H}%
L^{2}(T^{R}(X))_{t}$ by $\mathcal{H}L^{2}.$ The partial isometry theorem
(Theorem \ref{partial_isom.thm}) tells us that for $n\neq m,$ $\psi_{n}$ and
$\psi_{m}$ are orthogonal (but not orthonormal) as elements of $\mathcal{H}%
L^{2}.$

Suppose now that $F\in\mathcal{H}L^{2}$ and $\left\langle F,\psi
_{n}\right\rangle _{\mathcal{H}L^{2}}=0$ for all $n.$ By the holomorphic
change of variable (Theorem \ref{holo_change.thm}), with $\alpha$ as in
(\ref{alpha_form}), along with Theorem \ref{partial_inv2.thm}, we have%
\[
\left\langle F,\psi_{n}\right\rangle _{\mathcal{H}L^{2}}=\left\langle
F,\beta_{t,R}(-\Delta)\psi_{n}\right\rangle _{L^{2}(X)}=\beta_{t,R}%
(\lambda_{n})\left\langle F,\psi_{n}\right\rangle _{L^{2}(X)}.
\]
Here, in the second and third expressions, $F$ denotes the restriction of the
holomorphic function $F$ to $X.$ Now, as we have already remarked, it follows
from the partial isometry theorem (Theorem \ref{partial_isom.thm}) that
$\beta_{t,R}(\lambda_{n})$ is strictly positive for all $n,$ for all
$R<R_{2}.$ (See the discussion immediately after the statement of the
theorem.) Thus, if $F$ is orthogonal to each $\psi_{n}$ in $\mathcal{H}%
L^{2}(T^{R}(X))_{t},$ then the restriction of $F$ to $L^{2}(X)$ is orthogonal
to each $\psi_{n}.$ Since the $\psi_{n}$'s form an orthonormal basis for
$L^{2}(X),$ this tells us that the restriction of $F$ to $X$ is zero, from
which it follows that $F$ is zero on $T^{R}(X)$, because $X$ is a totally real
submanifold of maximal dimension in $T^{R}(X).$
\end{proof}

We now turn to the proof of the surjectivity theorem.

\begin{proof}
(Proof of Theorem \ref{surj.thm}) Suppose $F$ is as in Theorem \ref{surj.thm}.
The lemma tells us that for $R<\min(R_{1},R_{2})$ we can express $F$ as
\begin{equation}
F=\sum_{n=1}^{\infty}a_{n}\psi_{n}, \label{fsum}%
\end{equation}
with convergence in $\mathcal{H}L^{2}(T^{R}(X))_{t}.$ By a standard argument,
pointwise evaluation is continuous in $\mathcal{H}L^{2}(T^{R}(X))_{t}$ with
norm a locally bounded function of the point. It follows that the restriction
map from that space to $L^{2}(X)$ is a bounded operator. Thus, the same
expansion (\ref{fsum}) holds also in $L^{2}(X).$ This shows that the
coefficients in (\ref{fsum}) are independent of $R$ for a fixed holomorphic
function $F.$

We apply the partial isometry formula (Theorem \ref{partial_isom.thm}) with
$f_{1}=f_{2}=e^{t\lambda_{n}/2}\psi_{n},$ so that $F_{1}=F_{2}=\psi_{n}.$ This
tells us that the norm-squared of $\psi_{n}$ in $\mathcal{H}L^{2}%
(T^{R}(X))_{t}$ is $e^{t\lambda_{n}/2}\beta_{t,R}(\lambda_{n})$ (times the
norm-squared of $\psi_{n}$ in $L^{2}(X),$ which is 1). We conclude, then, that%
\begin{align}
G_{F}(R)  &  =\left\Vert F\right\Vert _{\mathcal{H}L^{2}(T^{R}(X))_{t}}%
^{2}=\sum_{n=1}^{\infty}\left\vert a_{n}\right\vert ^{2}e^{t\lambda_{n}%
/2}\beta_{t,R}(\lambda_{n})\nonumber\\
&  =\sum_{n=1}^{\infty}\left\vert a_{n}\right\vert ^{2}e^{t\left\vert
\rho\right\vert ^{2}}\int_{\substack{Y\in\mathbb{R}^{d} \\\left\vert
Y\right\vert \leq2R}}\exp\left(  \sqrt{\lambda_{n}-\left\vert \rho\right\vert
^{2}}~y_{1}\right)  \frac{e^{-\left\vert Y\right\vert ^{2}/4t}}{(4\pi
t)^{d/2}}~dY, \label{gf1}%
\end{align}
for all $R<\min(R_{1},R_{2}).$ We now split off the finite number of terms
where $\lambda_{n}<\left\vert \rho\right\vert ^{2}.$ Those terms have an
analytic continuation in $R$ given by the same expression. For the remaining
terms, the argument given in Section 7 of \cite{HM4} shows that if $G_{F}(R)$
is to have an analytic continuation in $R$ to $(0,\infty),$ it must be given
by (\ref{gf1}) for all $R.$

We now know that the analytic continuation of $G_{F}$ (which is assumed to
exist) is given by (\ref{gf1}) for all $R\in(0,\infty).$ To evaluate the limit
as $R\rightarrow\infty$ of $G_{F}(R),$ we use Dominated Convergence on the
finite number of terms with $\lambda_{n}<\left\vert \rho\right\vert ^{2}$ and
we use Monotone Convergence twice on the remaining terms to obtain%
\begin{align}
\lim_{R\rightarrow\infty}G_{F}(R)  &  =\sum_{n=1}^{\infty}\left\vert
a_{n}\right\vert ^{2}e^{t\left\vert \rho\right\vert ^{2}}\int_{Y\in
\mathbb{R}^{d}}\exp\left(  \sqrt{\lambda_{n}-\left\vert \rho\right\vert ^{2}%
}~y_{1}\right)  \frac{e^{-\left\vert Y\right\vert ^{2}/4t}}{(4\pi t)^{d/2}%
}~dY\nonumber\\
&  =\sum_{n=1}^{\infty}\left\vert a_{n}\right\vert ^{2}e^{t\lambda_{n}}.
\label{gf.lim}%
\end{align}

Since the limit of $G_{F}$ is assumed finite, we conclude that the right-hand
side of (\ref{gf.lim}) is finite. We may then define $f=\sum_{n=1}^{\infty
}a_{n}e^{t\lambda_{n}/2}\psi_{n}.$ The finiteness of (\ref{gf.lim}) gives
convergence of this series in $L^{2}(X)$ and we observe that $\left.
F\right\vert _{X}=\sum_{n=1}^{\infty}a_{n}\psi_{n}=e^{t\Delta/2}f.$ This
establishes Theorem \ref{surj.thm}.
\end{proof}


\begin{thebibliography}{99999}                                                                                            %


\bibitem[AG]{AG}D. N. Akhiezer and S. G. Gindikin, On Stein extensions of real
symmetric spaces, \textit{Math. Ann.} \textbf{286} (1990), 1--12.

\bibitem[Ba]{Ba1}V. Bargmann, On a Hilbert space of analytic functions and an
associated integral transform, \textit{Comm. Pure Appl. Math.} \textbf{14}
(1961), 187--214.

\bibitem[BHH]{BHH}D. Burns, S. Halverscheid, and R. Hind, The geometry of
Grauert tubes and complexification of symmetric spaces, \textit{Duke Math. J.}
\textbf{118} (2003), 465--491.

\bibitem[DOZ1]{DOZ1}M. Davidson, G. \'{O}lafsson, and G. Zhang, Laguerre
polynomials, restriction principle, and holomorphic representations of
$\mathrm{SL}(2,\mathbb{R})$, \textit{Acta Appl. Math.} \textbf{71} (2002), 261--277.

\bibitem[DOZ2]{DOZ2}M. Davidson, G. \'{O}lafsson, and G. Zhang, Laplace and
Segal-Bargmann transforms on Hermitian symmetric spaces and orthogonal
polynomials, \textit{J. Funct. Anal.} \textbf{204} (2003), 157--195.

\bibitem[DH]{DH1}B. K. Driver and B. C. Hall, Yang-Mills theory and the
Segal-Bargmann transform, \textit{Comm. Math. Phys.} \textbf{201} (1999), 249--290.

\bibitem[E]{E}L. D. \`{E}skin, Heat equation on Lie groups. (Russian) In: In
Memoriam: N. G. Chebotarev, 113--132, Izdat. Kazan. Univ., Kazan, Russia, 1964.

\bibitem[Far1]{Far1}J. Faraut, Formule de Gutzmer pour la complexification
d'un espace riemannien sym\'{e}trique. Harmonic analysis on complex
homogeneous domains and Lie groups (Rome, 2001). \textit{Atti Accad. Naz.
Lincei Cl. Sci. Fis. Mat. Natur. Rend. Lincei} \textit{(9) Mat. Appl.}
\textbf{13} (2002), no. 3-4, 233--241.

\bibitem[Far2]{Far2}J. Faraut, Analysis on the crown of a Riemannian symmetric
space. Lie groups and symmetric spaces, 99--110, Amer. Math. Soc. Transl. Ser.
2, 210, Amer. Math. Soc., Providence, RI, 2003.

\bibitem[Far3]{Far3}J. Faraut, Espaces hilbertiens invariants de fonctions
holomorphes. \textit{In} Analyse sur les groupes de Lie et th\'{e}orie des
repr\'{e}sentations (K\'{e}nitra, 1999), 101--167, S\'{e}min. Congr., 7, Soc.
Math. France, Paris, 2003

\bibitem[Fo]{Fo}G. B. Folland, \textquotedblleft Harmonic analysis in phase
space.\textquotedblright\ Annals of Mathematics Studies, 122. Princeton
University Press, Princeton, NJ, 1989.

\bibitem[FMMN1]{FMMN1}C. A. Florentino, P. Matias, J. M. Mour\~{a}o, and J. P.
Nunes, Geometric quantization, complex structures, and the coherent state
transform, \textit{J. Funct. Anal.} \textbf{221} (2005), 303--322.

\bibitem[FMMN2]{FMMN2}C. A. Florentino, P. Matias, J. M. Mour\~{a}o, and J. P.
Nunes, On the BKS pairing for K\"{a}hler quantizations of the cotangent bundle
of a Lie group, \textit{J. Funct. Anal.} \textbf{234} (2006), 180--198.

\bibitem[Ga]{Ga}R. Gangolli, Asymptotic behavior of spectra of compact
quotients of certain symmetric spaces, \textit{Acta Math.} \textbf{121}
(1968), 151--192.

\bibitem[GM]{GM}L. Gross and P. Malliavin, Hall's transform and the
Segal-Bargmann map, In: It\^{o}'s stochastic calculus and probability theory,
(M. Fukushima, N. Ikeda, H. Kunita, and S. Watanabe, Eds.), Springer-Verlag,
Berlin/New York, 1996, pp. 73-116.

\bibitem[GS1]{GStenz1}V. Guillemin and M. B. Stenzel, Grauert tubes and the
homogeneous Monge-Amp\`{e}re equation, \textit{J. Differential Geom.}
\textbf{34} (1991), 561--570.

\bibitem[GS2]{GStenz2}V. Guillemin and M. B. Stenzel, Grauert tubes and the
homogeneous Monge-Amp\`{e}re equation. II, \textit{J. Differential Geom.}
\textbf{35} (1992), 627--641.

\bibitem[H1]{H1}B. C. Hall, The Segal-Bargmann \textquotedblleft coherent
state\textquotedblright\ transform for compact Lie groups, \textit{J. Funct.
Anal.} \textbf{122} (1994), 103--151.

\bibitem[H2]{H2}B. C. Hall, The inverse Segal-Bargmann transform for compact
Lie groups, \textit{J. Funct. Anal.} \textbf{143} (1997), 98--116.

\bibitem[H3]{newform}B. C. Hall, A new form of the Segal-Bargmann transform
for Lie groups of compact type, \textit{Canad. J. Math.} \textbf{51} (1999), 816--834.

\bibitem[H4]{mexnotes}B. C. Hall, Holomorphic methods in analysis and
mathematical physics. In: First Summer School in Analysis and Mathematical
Physics (S. P\'{e}rez-Esteva and C. Villegas-Blas, Eds.), 1--59, Contemp.
Math., 260, Amer. Math. Soc., Providence, RI, 2000.

\bibitem[H5]{ymcoherent}B. C. Hall, Coherent states and the quantization of
(1+1)-dimensional Yang-Mills theory, \textit{Rev. Math. Phys.} \textbf{13}
(2001), 1281--1305.

\bibitem[H6]{bull}B. C. Hall, Harmonic analysis with respect to heat kernel
measure, \textit{Bull. Amer. Math. Soc. (N.S.)} \textbf{38} (2001), 43--78.

\bibitem[H7]{lpbounds}B. C. Hall, Bounds on the Segal-Bargmann transform of
$L^{p}$ functions. \textit{J. Fourier Anal. Appl.} \textbf{7} (2001), no. 6, 553--569.

\bibitem[H8]{geoquant}B. C. Hall, Geometric quantization and the generalized
Segal--Bargmann transform for Lie groups of compact type, \textit{Comm. Math.
Phys.} \textbf{226} (2002), 233-268.

\bibitem[H9]{ergodic}B. C. Hall, The Segal-Bargmann transform and the Gross
ergodicity theorem. In: Finite and infinite dimensional analysis in honor of
Leonard Gross (H.-H. Kuo and A. N. Sengupta, Eds.), 99--116, Contemp. Math.,
317, Amer. Math. Soc., Providence, RI, 2003.

\bibitem[H10]{range}B. C. Hall, The range of the heat operator. In: The
ubiquitous heat kernel (J. Jorgensen and L. Walling, Eds.) 203--231, Contemp.
Math., \textbf{398}, Amer. Math. Soc., Providence, RI, 2006.

\bibitem[HM1]{HM1}B. C. Hall and J. J. Mitchell, Coherent states on spheres,
\textit{J. Math. Phys.} \textbf{43} (2002), 1211--1236.

\bibitem[HM2]{HM3}B. C. Hall and J. J. Mitchell, The Segal-Bargmann transform
for noncompact symmetric spaces of the complex type, \textit{J. Funct. Anal.}
\textbf{227} (2005), 338--371.

\bibitem[HM3]{HM4}B. C. Hall and J. J. Mitchell, Isometry formula for the
Segal--Bargmann transform on a noncompact symmetric space of the complex type,
\textit{J. Funct. Anal.}, \textbf{254} (2008), 1575-1600.

\bibitem[HS]{HS}B. C. Hall and A. N. Sengupta, The Segal--Bargmann transform
for path-groups, \textit{J. Funct. Anal.} \textbf{152} (1998), 220-254.

\bibitem[He1]{He1}S. Helgason, \textquotedblleft Differential geometry, Lie
groups, and symmetric spaces.\textquotedblright\ Corrected reprint of the 1978
original. Graduate Studies in Mathematics, 34. American Mathematical Society,
Providence, RI, 2001.

\bibitem[He2]{He2}S. Helgason, \textquotedblleft Groups and geometric
analysis. Integral geometry, invariant differential operators, and spherical
functions.\textquotedblright\ Corrected reprint of the 1984 original.
Mathematical Surveys and Monographs, 83. American Mathematical Society,
Providence, RI, 2000.

\bibitem[He3]{He3}S. Helgason, \textquotedblleft Geometric analysis on
symmetric spaces.\textquotedblright\ Mathematical Surveys and Monographs, 39.
American Mathematical Society, Providence, RI, 1994.

\bibitem[Hu]{Hu}J. Huebschmann, Kirillov's character formula, the holomorphic
Peter-Weyl theorem, and the Blattner-Kostant-Sternberg pairing, \textit{J.
Geom. Phys.} \textbf{58} (2008), 833--848.

\bibitem[KOS]{KOS}B. Kr\"{o}tz, G. \'{O}lafsson, and R. Stanton, The image of
the heat kernel transform on Riemannian symmetric spaces of the non-compact
type, \textit{Int. Math. Res. Not.} 2005, 1307--1329.

\bibitem[KS1]{KS1}B. Kr\"{o}tz and R. J. Stanton, Holomorpic extensions of
representations: (I) automorphic functions, \textit{Ann. Math.} \textbf{159}
(2004), 641-724.

\bibitem[KS2]{KS2}B. Kr\"{o}tz and R. J. Stanton, Holomorphic extension of
representations: (II) geometry and harmonic analysis, \textit{Geom. Funct.
Anal.} \textbf{15} (2005), 190--245.

\bibitem[KTX1]{KTX}B. Kr\"{o}tz, S. Thangavelu, and Y. Xu, The heat kernel
transform for the Heisenberg group, \textit{J. Funct. Anal.} \textbf{225}
(2005), 301--336.

\bibitem[KTX2]{KTX2}B. Kr\"{o}tz, S. Thangavelu, and Y. Xu, Heat kernel
transform for nilmanifolds associated to the Heisenberg group, \textit{Rev.
Mat. Iberoam.} \textbf{24} (2008), 243--266.

\bibitem[Las]{Las}M. Lassalle, Series de Laurent des fonctions holomorphes
dans la complexication d'un espace symetrique compact, \textit{Ann. Scient.
\'{E}cole Norm. Sup.} 11 (1978), 167-210.

\bibitem[LGS]{LGS}\'{E}. Leichtnam, F. Golse, and M. B. Stenzel, Intrinsic
microlocal analysis and inversion formulae for the heat equation on compact
real-analytic Riemannian manifolds, \textit{Ann. Sci. \'{E}cole Norm. Sup.}
(4) \textbf{29} (1996), 669--736.

\bibitem[LS]{LS}L. Lempert and R. Sz\H{o}ke, Global solutions of the
homogeneous complex Monge-Amp\`{e}re equation and complex structures on the
tangent bundle of Riemannian manifolds, \textit{Math. Ann.} \textbf{290}
(1991), 689--712.

\bibitem[N]{N}E. Nelson, Analytic vectors, \textit{Ann. of Math.} (2)
\textbf{70} (1959) 572--615.

\bibitem[OO]{OO}G. \'{O}lafsson and B. \O rsted, Generalizations of the
Bargmann transform. In: Lie theory and its applications in physics (Clausthal,
1995), 3--14, World Sci. Publishing, River Edge, NJ, 1996.

\bibitem[OS1]{OS1}G. \'{O}lafsson and H. Schlichtkrull, The Segal--Bargmann
transform for the heat equation associated with root systems, \textit{Adv.
Math.} \textbf{208} (2007), 422-437.

\bibitem[OS2]{OS2}G. \'{O}lafsson and H. Schlichtkrull, Representation theory,
Radon transform and the heat equation on a Riemannian symmetric space. In:
Group representations, ergodic theory, and mathematical physics: a tribute to
George W. Mackey (R. S. Doran, C. C. Moore and R. J. Zimmer, Eds.), 315-344,
Amer. Math. Soc., RI, 2008.

\bibitem[Se1]{Se1}I. E. Segal, Mathematical problems of relativistic physics,
Chap.VI, In: Proceedings of the Summer Seminar, Boulder, Colorado, 1960, Vol.
II. (M. Kac, Ed.), Lectures in Applied Mathematics, American Math. Soc.,
Providence, Rhode Island, 1963.

\bibitem[Se2]{Se2}I. E. Segal, Mathematical characterization of the physical
vacuum for a linear Bose-Einstein field, \textit{Illinois J. Math.} \textbf{6}
(1962), 500-523

\bibitem[Se3]{Se3}I. E. Segal, The complex-wave representation of the free
Boson field. In: \textquotedblleft Topics in Functional
Analysis\textquotedblright\ (I. Gohberg and M. Kac, Eds.), Advances in
Mathematics Supplementary Studies, Vol. 3, Academic Press, New York, 1978.

\bibitem[SZ]{SZ}C. Sogge and S. Zelditch, Riemannian manifolds with maximal
eigenfunction growth, \textit{Duke Math. J.} \textbf{114} (2002), 387-437.

\bibitem[St1]{St1}M. B. Stenzel, The Segal-Bargmann transform on a symmetric
space of compact type, \textit{J. Funct. Anal.} \textbf{165} (1999), 44--58.

\bibitem[St2]{St2}M. B. Stenzel, An inversion formula for the Segal-Bargmann
transform on a symmetric space of non-compact type, \textit{J. Funct. Anal.}
\textbf{240} (2006), 592--608.

\bibitem[Sz1]{Sz1}R. Sz\H{o}ke, Complex structures on tangent bundles of
Riemannian manifolds, \textit{Math. Ann.} \textbf{291} (1991), 409--428.

\bibitem[Sz2]{Sz2}R. Sz\H{o}ke, Adapted complex structures and geometric
quantization, \textit{Nagoya Math. J.} \textbf{154} (1999), 171--183.

\bibitem[Sz3]{Sz3}R. Sz\H{o}ke, Involutive structures on the tangent bundle of
symmetric spaces. \textit{Math. Ann.} \textbf{319} (2001), 319--348.

\bibitem[Ty]{Ty}A. Tyurin, \textquotedblleft Quantization, classical and
quantum field theory and theta functions.\textquotedblright\ With a foreword
by Alexei Kokotov. CRM Monograph Series, 21. American Mathematical Society,
Providence, RI, 2003.

\bibitem[U]{U}H. Urakawa, The heat equation on compact Lie group,
\textit{Osaka J. Math.} \textbf{12} (1975), 285--297.

\bibitem[Wr]{Wr}K. K. Wren, Constrained quantisation and $\theta$-angles. II,
\textit{Nuclear Phys. B.} \textbf{521} (1998), 471--502.
\end{thebibliography}
\end{document}